\documentclass{elsarticle}

\usepackage{amsmath}
\usepackage{amssymb}
\usepackage{stmaryrd}
\usepackage{url}
\usepackage{enumerate}
\usepackage{subfigure}

\usepackage{pstricks,pst-node,pst-plot}
\newgray{dgray}{.5}
\newcmykcolor{dgreen}{1 0 .6 .3}
\newcmykcolor{dblue}{1 1 0 .3}
\newcmykcolor{lightgreen}{.25 0 .25 0}
\newcmykcolor{ultralightgreen}{.12 0 .12 0}

\newtheorem{theorem}{Theorem}[section]
\newtheorem{conjecture}[theorem]{Conjecture}
\newtheorem{lemma}[theorem]{Lemma}

\newdefinition{definition}[theorem]{Definition }
\newdefinition{example}[theorem]{Example }

\newcommand{\summary}[1]{\textrm{\textbf{\textup{#1}}}}

\newcommand{\bigO}{\mathop{\mathrm{O}}\nolimits}

\newproof{proof}{Proof}

\providecommand*{\Nset}{\mathbb{N}}            
\providecommand*{\Zset}{\mathbb{Z}}            
\providecommand*{\Qset}{\mathbb{Q}}            
\providecommand*{\Rset}{\mathbb{R}}            
\providecommand*{\nonnegRset}{\mathbb{R}_{\scriptscriptstyle{+}}}
\providecommand*{\CPset}{\mathbb{CP}}          
\providecommand*{\Pset}{\mathbb{P}}            

\providecommand*{\NNCBox}{\mathbb{B}}          
\providecommand*{\NNCInterval}{\mathbb{I}}     

\providecommand*{\Aset}{\mathbb{A}}
\providecommand*{\Aiset}[1]{\mathbb{A}(#1)}

\newcommand*{\cA}{\ensuremath{\mathcal{A}}}
\newcommand*{\cB}{\ensuremath{\mathcal{B}}}
\newcommand*{\cC}{\ensuremath{\mathcal{C}}}

\newcommand*{\cG}{\ensuremath{\mathcal{G}}}
\newcommand*{\cH}{\ensuremath{\mathcal{H}}}

\newcommand*{\cN}{\ensuremath{\mathcal{N}}}

\newcommand*{\cP}{\ensuremath{\mathcal{P}}}
\newcommand*{\cQ}{\ensuremath{\mathcal{Q}}}

\newcommand*{\rc}{\ensuremath{\mathrm{c}}}

\newcommand*{\rt}{\ensuremath{\mathrm{t}}}

\renewcommand{\emptyset}{\mathord{\varnothing}}

\newcommand*{\sseq}{\subseteq}

\newcommand*{\sslt}{\subset}
\newcommand*{\Sseq}{\supseteq}

\newcommand{\Nsseq}{\nsubseteq}

\newcommand*{\sqsseq}{\sqsubseteq}

\newcommand*{\union}{\cup}
\newcommand*{\bigunion}{\bigcup}
\newcommand*{\biginters}{\bigcap}
\newcommand*{\inters}{\cap}
\newcommand*{\setdiff}{\setminus}

\newcommand{\sset}[2]{{\renewcommand{\arraystretch}{1.2}
                      \left\{\,#1 \,\left|\,
                               \begin{array}{@{}l@{}}#2\end{array}
                      \right.   \,\right\}}}

\newcommand*{\fund}[3]{\mathord{#1}\colon#2\rightarrow#3}

\newcommand{\st}{\mathrel{.}}
\newcommand{\itc}{\mathrel{:}}

\newcommand*{\vect}[1]{\mathbf{#1}}

\newcommand*{\widen}{\mathbin{\nabla}}

\newcommand*{\bigland}{\mathop{\bigwedge}}

\newcommand{\defeq}{\mathrel{\mathord{:}\mathord{=}}}

\newcommand*{\proj}{\mathop{\pi}\nolimits}
\newcommand*{\octproj}{\mathop{\tilde\pi}\nolimits}

\newcommand{\card}{\mathop{\#}\nolimits}

\newcommand*{\Vpm}{\cN}

\newcommand*{\con}{\mathop{\mathrm{con}}\nolimits}
\newcommand*{\gen}{\mathop{\mathrm{gen}}\nolimits}
\newcommand*{\genleq}{\sqsubseteq}
\newcommand*{\genlt}{\sqsubset}

\renewcommand*{\bar}[1]{\overline{#1}}
\newcommand*{\bari}{\overline{\imath}}
\newcommand*{\barj}{\overline{\jmath}}

\newcommand*{\bark}{\overline{k}}
\newcommand*{\barl}{\overline{\ell}}

\newcommand*{\Graphs}{\mathbb{G}}

\newcommand*{\graphleq}{\unlhd}

\newcommand*{\graphglb}{\sqcap}

\newcommand*{\graphlub}{\sqcup}
\newcommand*{\biggraphlub}{\bigsqcup}
\newcommand*{\Octgraphs}{\mathbb{O}}

\renewcommand*{\path}{\theta}
\newcommand*{\pathconc}{\mathrel{::}}

\newcommand*{\BDshapes}{\mathbb{BD}}
\newcommand*{\IBDshapes}{\mathbb{BD}^\Zset}
\newcommand*{\bd}{\mathrm{bd}}
\newcommand*{\Octshapes}{\mathbb{OCT}}
\newcommand*{\IOctshapes}{\mathbb{OCT}^\Zset}
\newcommand*{\oct}{\mathrm{oct}}

\newcommand*{\closure}{\mathop{\mathrm{closure}}\nolimits}
\newcommand*{\strongclosure}{\mathop{\text{\rm S-closure}}\nolimits}
\newcommand*{\tightclosure}{\mathop{\text{\rm T-closure}}\nolimits}

\newcommand*{\rays}{\mathop{\mathrm{rays}}\nolimits}

\newcommand*{\transpose}{{\scriptscriptstyle\mathrm{T}}}

\newcommand*{\topclosure}{\mathop{\mathbb{C}}\nolimits}

\newcommand*{\relop}{\mathrel{\bowtie}}
\newcommand*{\relops}[1]{\mathrel{\bowtie_{#1}}}

\newcommand{\just}[1]{\text{[#1]}}

\begin{document}

\begin{frontmatter}

\title{Exact Join Detection
       for Convex Polyhedra \\
       and Other Numerical Abstractions\tnoteref{th}}

\tnotetext[th]{This work has been partly supported by PRIN project
``AIDA2007 --- Abstract Interpretation Design and Applications,''
and by EPSRC project
``EP/G00109X/1 --- Static Analysis Tools for Certifying and Debugging Programs.''
}

\author{Roberto Bagnara}
\ead{bagnara@cs.unipr.it}
\author{Patricia M. Hill}
\ead{hill@cs.unipr.it}
\author{Enea Zaffanella}
\ead{zaffanella@cs.unipr.it}

\address{Department of Mathematics, University of Parma, Italy}

\begin{abstract}
Deciding whether the union of two convex polyhedra
is itself a convex polyhedron is a basic problem in polyhedral computations;
having important applications in the field of constrained control
and in the synthesis, analysis, verification and optimization
of hardware and software systems.  In such application fields
though, general convex polyhedra are just one among many, so-called,
\emph{numerical abstractions}, which range from restricted families
of (not necessarily closed) convex polyhedra to non-convex geometrical
objects.  We thus tackle the problem from an abstract point of view:
for a wide range of numerical abstractions that can be modeled
as bounded join-semilattices ---that is, partial orders where any
finite set of elements has a least upper bound---,
we show necessary and
sufficient conditions for the equivalence between the lattice-theoretic join
and the set-theoretic union.  For the case of closed convex
polyhedra ---which, as far as we know, is the only one already
studied in the literature--- we improve upon the state-of-the-art
by providing a new algorithm with a better worst-case complexity.
The results and algorithms presented for the other numerical
abstractions are new to this paper.  All the algorithms have
been implemented, experimentally validated, and made available
in the Parma Polyhedra Library.
\end{abstract}

\begin{keyword}
polyhedron, union, convexity, abstract interpretation,
numerical abstraction, powerset domain.
\end{keyword}

\end{frontmatter}

\section{Introduction}
\label{sec:introduction}

For $n \in \Nset$, let $\mathbb{D}_n \sslt \wp(\Rset^n)$ be a set of
finitely-representable sets such that $(\mathbb{D}_n, \mathord{\sseq})$
is a bounded join-semilattice, that is, a minimum element exists
as well as the least upper bound
for all $D_1, D_2 \in \mathbb{D}_n$.
Such a least upper bound ---let us denote it by $D_1 \uplus D_2$
and call it the \emph{join} of $D_1$ and $D_2$---
is, of course, not guaranteed to be equal to $D_1 \union D_2$.
More generally, we refer to the problem of deciding, for each
finite set $\{ D_1, \ldots, D_k \} \sseq \mathbb{D}_n$,
whether $\biguplus_{i=1}^k D_i = \bigunion_{i=1}^k D_i$
as the \emph{exact join detection} problem.

Examples of $\mathbb{D}_n$ include $n$-dimensional convex polyhedra,
either topologically closed or not necessarily so, restricted families
of polyhedra characterized by interesting algorithmic complexities
---such as \emph{bounded-difference} and \emph{octagonal shapes}---,
Cartesian products of some families of intervals, and other ``box-like''
geometric objects where the intervals can have ``holes'' (for instance,
Cartesian products of \emph{modulo intervals}
\cite{NakanishiF01,NakanishiJPF99} fall in this category).
All these \emph{numerical abstractions} allow to conveniently represent
or approximate the constraints arising in constrained control
(see, e.g., \cite{Jones05th})
and, more generally, in the synthesis, analysis, verification and optimization
of hardware and software systems (see, e.g., \cite{BagnaraHZ09TCS}).

The restrictions implied by convexity and/or by the ``shapes'' of the
geometric objects in $\mathbb{D}_n$ are sometimes inappropriate for the
application at hand.  In these cases, one possibility is to consider
finite sets of elements of $\mathbb{D}_n$.  For instance, many applications
in the field of hardware/software verification use constructions like
the \emph{finite powerset domain} of \cite{Bagnara98SCP}:
this is a special case of \emph{disjunctive completion}~\cite{CousotC79},
where disjunctions are implemented by maintaining an explicit (hence finite)
and \emph{non-redundant} collection of elements of $\mathbb{D}_n$.
Non-redundancy means that a collection is made of maximal elements
with respect to subset inclusion, so that no element is contained in
another element in the collection.
The finite powerset and similar constructions are such that
$Q_1 = \{ D_1, \dots, D_{h-1}, D_h, \dots, D_k \}$
and $Q_2 = \{ D_1, \dots, D_{h-1}, D \}$ are two different representations
for the same set, if $\bigunion_{i=h}^k D_i = \biguplus_{i=h}^k D_i = D$.
The latter representation is clearly more desirable, and not just
because ---being more compact--- it results in a better efficiency of
all the involved algorithms.
In the field of control engineering, the ability of efficiently simplifying
$Q_1$ into $Q_2$ can be used to reduce the complexity of the solution
to optimal control problems, thus allowing for the synthesis of cheaper
control hardware \cite{BemporadMDP02,Torrisi03th}.
Similarly, the simplification of $Q_1$ into $Q_2$ can lead to improvements
in loop optimizations obtained by automatic code generators such as
CLooG \cite{Bastoul04}.  In the same application area, this simplification
allows for a reduction in the complexity of array data-flow analysis and
for a simplification of
\emph{quasi-affine selection trees} (QUASTs).  In loop optimization,
dependencies between program statements are modeled by parametric linear
systems, whose solutions can be represented by QUASTs and computed by
tools like PIP \cite{Feautrier88}, which, however, can generate non-simplified
QUASTs.  These can be simplified efficiently provided there is an efficient
procedure for deciding the exact join property.
Another application of exact join detection is the computation of
under-approximations, which are useful, in particular, for the
approximation of contra-variant operators such as set-theoretic difference.
In fact, when the join is exact it is a safe under-approximation of the
union.
The exact join detection procedure can also be used as a preprocessing
step for the \emph{extended convex hull} problem\footnote{This is the problem
of computing a minimal set of constraints describing the convex hull
of the union of $k$ polytopes, each described by a set of
constraints.} \cite{FukudaLL01}.
Another important application of exact join detection
comes from the field of static analysis
via \emph{abstract interpretation} \cite{CousotC77,CousotC79}.
In abstract interpretation, static analysis is usually conducted by performing
a fixpoint computation.  Suppose we use the finite powerset domain
$\bigl(\wp_\mathrm{fn}(\mathbb{D}_n), \sqsseq, \emptyset, \sqcup\bigr)$:
this is the bounded join-semilattice of the \emph{finite} and
\emph{non-redundant}
subsets of $\mathbb{D}_n$ ordered by the relation given,
for each $Q_1, Q_2 \in \wp_\mathrm{fn}(\mathbb{D}_n)$, by
\[
  Q_1 \sqsseq Q_2
    \quad\iff\quad
      \forall D_1 \in Q_1 \itc
        \exists D_2 \in Q_2 \st
          D_1 \sseq D_2,
\]
and `$\mathord{\sqcup}$' is the least upper bound (join) operator induced
by `$\mathord{\sqsseq}$' \cite{BagnaraHZ06STTT}.
The system under analysis is approximated by a monotonic
(so called) \emph{abstract semantic function}
$\fund{\cA}{\wp_\mathrm{fn}(\mathbb{D}_n)}{\wp_\mathrm{fn}(\mathbb{D}_n)}$,
and the limit of the ascending chain given by $\cA$'s iterates,
\begin{equation}
\label{eq:unwidened-sequence}
  \cA^0(\emptyset), \cA^1(\emptyset), \cA^2(\emptyset), \dots,
\end{equation}
is, by construction, a sound approximation of the analyzed system's behavior.
Since $\wp_\mathrm{fn}(\mathbb{D}_n)$ has infinite ascending chains,
the standard abstract iteration sequence \eqref{eq:unwidened-sequence}
may converge very slowly or fail to converge altogether.
For this reason, a \emph{widening operator}
$\fund{\widen}{\wp_\mathrm{fn}(\mathbb{D}_n)^2}{\wp_\mathrm{fn}(\mathbb{D}_n)}$
is introduced.  This ensures that the sequence
\begin{equation}
\label{eq:widened-sequence}
  \cB^0(\emptyset), \cB^1(\emptyset), \cB^2(\emptyset), \dots.
\end{equation}
where, for each $Q \in \wp_\mathrm{fn}(\mathbb{D}_n)$,
$\cB(Q) \defeq Q \widen \bigl(Q \sqcup \cA(Q)\bigr)$,
is ultimately stationary and that the (finitely computable) fixpoint
of $\cB$ is a post-fixpoint of $\cA$, i.e., a sound approximation
of the behavior of the system under consideration.
In \cite{BagnaraHZ06STTT} three generic widening methodologies
are presented for finite powerset abstract domains.
A common trait of these methodologies is given by the fact that the
precision/efficiency trade-off of the resulting widening can be greatly
improved if domain elements are ``pairwise merged'' or even
``fully merged.''  Let the cardinality of a finite set $S$ be denoted by
$\card S$.  An element $Q = \{ D_1, \dots, D_h \}$ of
$\wp_\mathrm{fn}(\mathbb{D}_n)$ is said to be \emph{pairwise merged}
if, for each $R \sseq Q$, $\card R = 2$ implies
$\bigunion R \neq \biguplus R$;
the notion of being \emph{fully merged} is obtained by replacing
$\card R = 2$ with $\card R \geq 2$ in the above.

In this paper, we tackle the problem of exact join detection
for all the numerical abstractions that are in widespread use at the
time of writing.\footnote{Since numerical abstractions are so critical
in the field of hardware and software analysis and verification,
new ones are proposed on a regular basis.}
This problem has been studied for convex polyhedra in~\cite{BemporadFT01}.
We are not aware of any work that addresses the problem for other numerical
abstractions.

In \cite{BemporadFT01} the authors provide theoretical results and
algorithms for the exact join detection problem applied to
a pair of topologically closed convex polyhedra.
Three different specializations of the problem are considered,
depending on the chosen representation for the input polyhedra:
H-polyhedra, described by constraints (half-spaces);
V-polyhedra, described by generators (vertices); and
VH-polyhedra, described by both constraints and generators.%
\footnote{The algorithms in~\cite{BemporadFT01} for the V and VH
representations only consider the case of \emph{bounded} polyhedra,
i.e., polytopes; the extension to the unbounded case can be found
in~\cite{BemporadFT00TR}.}
The algorithms for the H and V representations, which are based on
Linear Programming techniques, enjoy a polynomial worst-case complexity
bound; the algorithm for VH-polyhedra achieves a better,
strongly polynomial bound.
For the H-polyhedra case only, it is also shown how the algorithm
can be generalized to more than two input polyhedra.
An improved theoretical result for the case of
more than two V-polytopes is stated in~\cite{BaranyF05}.

The first contribution of the present paper is a theoretical result
for the VH-polyhedra case, leading to the specification of a
new algorithm improving upon the worst-case complexity bound
of~\cite{BemporadFT00TR}.

The second contribution is constituted by original results and
algorithms concerning the exact join detection problem for
the other numerical abstractions.
For those that are restricted classes of topologically closed convex
polyhedra, one could of course use the same algorithms used for the
general case, but the efficiency would be poor.
Consider that the applications of finite powersets of numerical abstractions
range between two extremes:
\begin{itemize}
\item
those using small-cardinality powersets of complex abstractions
such as general polyhedra (see, for instance \cite{BultanGP99});
\item
those using large-cardinality powersets of simple abstractions
(for instance, verification tasks like the one described
in \cite{FrehseKR06}, can be tackled this way).
\end{itemize}
So, in general, the simplicity of the abstractions is countered by
their average number in the powersets.  It is thus clear that
specialized, efficient algorithms are needed for all numerical abstractions.
In this paper we present algorithms, each backed with the corresponding
correctness result, for the following numerical abstractions:
not necessarily closed convex polyhedra,
``box-like'' geometric objects;
rational (resp., integer) bounded difference shapes; and
rational (resp., integer) octagonal shapes.

The plan of the paper is as follows.
In Section~\ref{sec:preliminaries}, we introduce the required
technical notation and terminology.
In Section~\ref{sec:polyhedra},
we discuss the results and algorithms for convex polyhedra.
The specialized results for boxes, bounded difference shapes and
octagonal shapes are provided in Sections~\ref{sec:boxes},
\ref{sec:bd-shapes} and~\ref{sec:oct-shapes}, respectively.
Section~\ref{sec:conclusion} concludes.

\section{Preliminaries}
\label{sec:preliminaries}

The set of non-negative reals is denoted by $\nonnegRset$.
In the present paper, all topological arguments refer to the Euclidean
topological space $\Rset^n$, for any positive integer $n$.
If $S \sseq \Rset^n$, then
the \emph{topological closure} of $S$ is defined as
\(
  \topclosure(S)
    \defeq
      \biginters
        \{\,
          C \sseq \Rset^n
        \mid
          \text{$S \sseq C$ and $C$ is closed}
        \,\}
\).

For each $i \in \{1, \ldots, n\}$, $v_i$ denotes the $i$-th component
of the (column) vector $\vect{v} \in \Rset^n$; the projection on
space dimension $i$ for a set $S \sseq \Rset^n$ is denoted by
\(
  \proj_i(S) \defeq \{\, v_i \in \Rset \mid \vect{v} \in S \,\}
\).
We denote by $\vect{0}$ the vector of $\Rset^n$
having all components equal to zero.
A vector $\vect{v} \in \Rset^n$ can also be interpreted as a
matrix in $\Rset^{n \times 1}$ and manipulated accordingly with the
usual definitions for addition, multiplication (both by a scalar and
by another matrix), and transposition, which is denoted by
$\vect{v}^\transpose$.  The \emph{scalar product} of
$\vect{v},\vect{w} \in \Rset^n$, denoted
$\langle \vect{v}, \vect{w}\rangle$, is the real number
\(
  \vect{v}^\transpose \vect{w} = \sum_{i=1}^{n} v_i w_i
\).

For any relational operator
$\mathord{\relop} \in \{ =, \leq, \geq, <, > \}$, we write
$\vect{v} \relop \vect{w}$ to denote the conjunctive
proposition $\bigland_{i=1}^{n} (v_i \relop w_i)$. Moreover,
$\vect{v} \neq \vect{w}$ denotes the proposition
$\neg (\vect{v} = \vect{w})$.
We occasionally use the convenient notation
$a \relops{1} b \relops{2} c$ to denote the conjunction
$a \relops{1} b \land b \relops{2} c$ and do not distinguish
conjunctions of propositions from sets of propositions.

\subsection{Topologically Closed Convex Polyhedra}

For each vector $\vect{a} \in \Rset^n$
and scalar $b \in \Rset$, where $\vect{a} \neq \vect{0}$,
the linear non-strict inequality constraint
$\beta = \bigl( \langle \vect{a}, \vect{x} \rangle \leq b \bigr)$
defines a topologically closed affine half-space of $\Rset^n$.
The linear equality constraint
$\langle \vect{a}, \vect{x} \rangle = b$
defines an affine hyperplane.
A topologically closed convex polyhedron is usually described
as a finite system of linear equality and non-strict inequality constraints.
Theoretically speaking, it is simpler to express
each equality constraint as the intersection of the two half-spaces
$\langle \vect{a}, \vect{x} \rangle \leq b$ and
$\langle -\vect{a}, \vect{x} \rangle \leq -b$.
We do not distinguish between syntactically different constraints
defining the same affine half-space so that, e.g.,
$x \leq 2$ and $2x \leq 4$ are considered to be the same constraint.


We write $\con(\cC)$ to denote the polyhedron $\cP \sseq \Rset^n$
described by the finite \emph{constraint system} $\cC$.
Formally, we define
\[
  \con(\cC)
    \defeq
      \Bigl\{\,
        \vect{p} \in \Rset^n
      \Bigm|
        \forall \beta = \bigl(
                          \langle \vect{a}, \vect{x} \rangle \leq b
                        \bigr) \in \cC
          \itc
            \langle \vect{a}, \vect{p} \rangle \leq b
      \,\Bigr\}.
\]
The function `$\mathord{\con}$' enjoys an anti-monotonicity property,
meaning that $\cC_1 \sseq \cC_2$ implies $\con(\cC_1) \Sseq \con(\cC_2)$.

Alternatively, the definition of a topologically closed convex polyhedron
can be based on some of its geometric features.
A vector $\vect{r} \in \Rset^n$ such that $\vect{r} \neq \vect{0}$
is a \emph{ray} (or \emph{direction of infinity}) of a non-empty polyhedron
$\cP \sseq \Rset^n$ if, for every point $\vect{p} \in \cP$
and every non-negative scalar $\rho \in \nonnegRset$,
we have $\vect{p} + \rho \vect{r} \in \cP$;
the set of all the rays of a polyhedron $\cP$ is denoted by $\rays(\cP)$.
A vector $\vect{l} \in \Rset^n$ is a \emph{line} of $\cP$
if both $\vect{l}$ and $-\vect{l}$ are rays of $\cP$.
The empty polyhedron has no rays and no lines.
As was the case for equality constraints, the theory
can dispense with the use of lines by using the corresponding
pair of rays.
Moreover, when vectors are used to denote rays, no distinction is
made between different vectors having the same direction so that,
e.g., $\vect{r}_1 = (1, 3)^\transpose$ and $\vect{r}_2 = (2, 6)^\transpose$
are considered to be the same ray in $\Rset^2$.
The following theorem is a simple consequence of
well-known theorems by Minkowski and Weyl~\cite{StoerW70}.

\begin{theorem}
\label{thm:minkowski-weyl}
The set $\cP \sseq \Rset^n$ is a closed polyhedron if and only if
there exist finite sets $R, P \sseq \Rset^n$
of cardinality $r$ and $p$, respectively,
such that $\vect{0} \notin R$ and
\[
  \cP = \gen\bigl( (R, P) \bigr)
      \defeq
        \biggl\{\,
          R \vect{\rho} + P \vect{\sigma} \in \Rset^n
        \biggm|
          \vect{\rho} \in \nonnegRset^r,
          \vect{\sigma} \in \nonnegRset^p,
          \sum_{i=1}^p \sigma_i = 1
        \,\biggr\}.
\]
\end{theorem}
When $\cP \neq \emptyset$, we say that $\cP$ is described by
the \emph{generator system} $\cG = (R, P)$.
In particular, the vectors of $R$ and $P$ are rays and points
of $\cP$, respectively.
Thus, each point of the generated polyhedron is obtained by
adding a non-negative combination of the rays in $R$ and a convex
combination of the points in $P$.
Informally speaking, if no ``supporting point''
is provided then an empty polyhedron is obtained;
formally, $\cP = \emptyset$ if and only if $P = \emptyset$.
By convention, the empty system
(i.e., the system with $R = \emptyset$ and $P = \emptyset$)
is the only generator system for the empty polyhedron.
We define a partial order relation `$\genleq$' on generator systems,
which is the component-wise extension of set inclusion.
Namely, for any generator systems
$\cG_1 = (R_1, P_1)$ and $\cG_2 = (R_2, P_2)$, we have
$\cG_1 \genleq \cG_2$ if and only if $R_1 \sseq R_2$ and
$P_1 \sseq P_2$;
if, in addition, $\cG_1 \neq \cG_2$,
we write $\cG_1 \genlt \cG_2$.
The function `$\mathord{\gen}$' enjoys a monotonicity property,
as $\cG_1 \genleq \cG_2$ implies $\gen(\cG_1) \sseq \gen(\cG_2)$.

The vector $\vect{v} \in \cP$ is an \emph{extreme point}
(or \emph{vertex}) of the polyhedron $\cP$ if it cannot be expressed
as a convex combination of some other points of $\cP$.
Similarly, $\vect{r} \in \rays(\cP)$ is an \emph{extreme ray} of $\cP$
if it cannot be expressed as a non-negative combination of some other
rays of $\cP$.
It is worth stressing that, in general, the vectors in $R$ and $P$
are not the extreme rays and the vertices
of the polyhedron:
for instance, any half-space of $\Rset^2$ has
two extreme rays and no vertices, but any generator system
describing it will contain at least three rays and one point.

The combination of the two approaches outlined above is the basis
of the double description method due to Motzkin et al.~\cite{MotzkinRTT53},
which exploits the duality principle to compute each representation
starting from the other one, possibly minimizing both descriptions.
Clever implementations of this \emph{conversion} procedure,
such as those based on the extension by Le~Verge~\cite{LeVerge92} of
Chernikova's algorithms~\cite{Chernikova64,Chernikova65,Chernikova68},
are the starting points for the development of software libraries
based on the double description method.
While being characterized by a worst-case computational cost
that is exponential in the size of the input, these algorithms
turn out to be practically useful for the purposes of many
applications in the context of static analysis.

We denote by $\CPset_n$ the set of all topologically closed polyhedra
in $\Rset^n$, which is partially ordered by subset inclusion to form
a non-complete lattice; the finitary greatest lower bound operator
corresponds to intersection; the finitary least upper bound operator,
denoted by `$\mathord{\uplus}$', corresponds to the convex polyhedral hull.
Observe that if, for each $i \in \{1, 2\}$,
$\cP_i = \gen\bigl((R_i, P_i)\bigr)$,
then the convex polyhedral hull is
$\cP_1 \uplus \cP_2 = \gen\bigl((R_1 \union R_2, P_1 \union P_2)\bigr)$.

\subsection{Not Necessarily Closed Convex Polyhedra}

The linear strict inequality constraint
$\beta = \bigl( \langle \vect{a}, \vect{x} \rangle > b \bigr)$
defines a topologically open affine half-space of $\Rset^n$.
A not necessarily closed (NNC) convex polyhedron is defined by
a finite system of strict and non-strict inequality constraints.
Since by using lines, rays and points we can only represent
topologically closed polyhedra, the key step for a parametric description
of NNC polyhedra is the introduction of a new kind of generator
called a \emph{closure point}~\cite{BagnaraHZ05FAC}.
\begin{definition}
\label{def:closure-point}
\summary{(Closure point.)}
A vector $\vect{c} \in \Rset^n$ is a \emph{closure point}
of $S \sseq \Rset^n$ if and only if $\vect{c} \in \topclosure(S)$.
\end{definition}
For a non-empty NNC polyhedron $\cP \sseq \Rset^n$, a vector
$\vect{c} \in \Rset^n$ is a closure point of $\cP$ if and only if
\(
  \sigma \vect{p} + (1 - \sigma) \vect{c} \in \cP
\)
for every point $\vect{p} \in \cP$
and every $\sigma \in \Rset$ such that $0 < \sigma < 1$.
By excluding the case when $\sigma = 0$, $\vect{c}$ is not forced
to be in $\cP$.

The following theorem taken from~\cite{BagnaraHZ05FAC} is a generalisation of
Theorem~\ref{thm:minkowski-weyl} to NNC polyhedra.
\begin{theorem}
\label{thm:NNC-minkowski-weyl}
The set $\cP \sseq \Rset^n$ is an NNC polyhedron if and only if
there exist finite sets $R, P, C \sseq \Rset^n$
of cardinality $r$, $p$ and $c$, respectively,
such that $\vect{0} \notin R$ and
\[
  \cP
    = \gen\bigl( (R, P, C) \bigr)
    \defeq
      \sset{
        R \vect{\rho} + P \vect{\sigma} + C \vect{\tau} \in \Rset^n
      }{
        \vect{\rho} \in \nonnegRset^r,
        \vect{\sigma} \in \nonnegRset^p, \vect{\sigma} \neq \vect{0}, \\
        \vect{\tau} \in \nonnegRset^c, \\
        \sum_{i=1}^p \sigma_i + \sum_{i=1}^c \tau_i = 1
      }.
\]
\end{theorem}
When $\cP \neq \emptyset$, we say that $\cP$ is described by
the \emph{extended generator system} $\cG = (R, P, C)$.
As was the case for closed polyhedra, the vectors in $R$ and $P$
are rays and points of $\cP$, respectively.
The condition $\vect{\sigma} \neq \vect{0}$ ensures that
at least one of the points of $P$ plays an active role
in any convex combination of the vectors of $P$ and $C$.
The vectors of $C$ are closure points of $\cP$.
Since both rays and closure points need a supporting point, we have
$\cP = \emptyset$ if and only if $P = \emptyset$.
The partial order relation `$\genleq$' on generator systems is
easily extended to also take into account the closure points component,
so that the overloading of the function `$\mathord{\gen}$' still satisfies
the monotonicity property.

The set of all NNC polyhedra in $\Rset^n$, denoted $\Pset_n$,
is again a non-complete lattice partially ordered by subset inclusion,
having $\CPset_n$ as a sublattice.
As for the set of closed polyhedra $\CPset_n$,
the finitary greatest lower bound operator
corresponds to intersection; the finitary least upper bound operator,
again denoted by `$\mathord{\uplus}$', corresponds to the
not necessarily closed convex polyhedral hull.
Observe that if, for each $i \in \{1, 2\}$,
$\cP_i = \gen\bigl((R_i, P_i, C_i)\bigr)$,
then the convex polyhedral hull is
\(
  \cP_1 \uplus \cP_2
    = \gen\bigl((R_1 \union R_2, P_1 \union P_2, C_1 \union C_2)\bigr)
\).

\subsection{Subsumption and Saturation}

A point (resp., ray, closure point) $\vect{v} \in \Rset^n$
is said to be \emph{subsumed} by a polyhedron $\cP$ if and only if
$\vect{v}$ is a point (resp., ray, closure point) of $\cP$.

A (closure) point $\vect{p} \in \Rset^n$ is said to \emph{saturate}
a constraint
\(
  \beta = \bigl(
            \langle \vect{a}, \vect{x} \rangle \relop b
          \bigr)
\),
where $\mathord{\relop} \in \{ =, \leq, \geq, <, > \}$,
if and only if
$\langle \vect{a}, \vect{p} \rangle = b$;
a ray $\vect{r} \in \Rset^n$ is said to saturate
the same constraint $\beta$ if and only if
$\langle \vect{a}, \vect{r} \rangle = 0$.

\section{Exact Join Detection for Convex Polyhedra}
\label{sec:polyhedra}
In this section, we provide results for the exact join detection problem
for convex polyhedra.
Here we just consider the case when a double description representation is
available; that is, in the proposed methods, we exploit both the
constraint and the generator descriptions of the polyhedra.

\subsection{Exact Join Detection for Topologically Closed Polyhedra}

First we consider the exact join detection problem for closed
polyhedra since, in this case, given any two closed polyhedra
$\cP_1, \cP_2 \in \CPset_n$,
we have that $\cP_1 \union \cP_2$ is convex if and only
if $\cP_1 \uplus \cP_2 = \cP_1 \union \cP_2$.
Before stating and proving the main result for this section, we
present the following lemma that establishes some simple conditions
that will ensure the union of two closed polyhedra is not convex.

\begin{lemma}
\label{lem:union-of-closed-polyhedra-is-not-convex}
Let $\cP_1, \cP_2 \in \CPset_n$ be
topologically closed non-empty polyhedra.
Suppose there exist a constraint $\beta$ and
a vector $\vect{p}$ such that
\textup{(1)} $\vect{p}$ saturates $\beta$,
\textup{(2)} $\beta$ is satisfied by $\cP_1$ but violated by $\cP_2$, and
\textup{(3)} $\vect{p} \in \cP_1 \setdiff \cP_2$.
Then, $\cP_1 \union \cP_2$ is not convex.
\end{lemma}

\begin{proof}
(See also Figure~\ref{subfig:lemmas:closed}.)
By (2), there exists a point $\vect{p}_2 \in \cP_2$ that violates $\beta$.
Consider the closed line segment $s \defeq [\vect{p}, \vect{p}_2]$;
by (1), the one and only point on this segment
that satisfies $\beta$ is $\vect{p}$;
by (3),  $\vect{p} \in \cP_1$ so that $s \sseq \cP_1 \uplus \cP_2$.
Also by (3), $\vect{p} \notin \cP_2$, so that
there exists a non-strict constraint $\beta_2$
that is satisfied by $\cP_2$ but violated by $\vect{p}$.
Since $\vect{p}_2 \in \cP_2$,
there exists a vector $\vect{q} \in s$ that saturates $\beta_2$
and $\vect{q} \neq \vect{p}$.
It follows that the open line segment
$s_1 \defeq (\vect{p}, \vect{q})$ is non-empty
and every point in $s_1$ violates both $\beta$ and $\beta_2$;
hence $s_1 \inters \cP_1 = s_1 \inters \cP_2 = \emptyset$.
However, by construction,
\[
  (\vect{p}, \vect{q})
    \sslt [\vect{p}, \vect{p}_2]
      \sseq \cP_1 \uplus \cP_2,
\]
so that $\cP_1 \uplus \cP_2 \neq \cP_1 \union \cP_2$.
Therefore $\cP_1 \union \cP_2$ is not convex.
\qed
\end{proof}

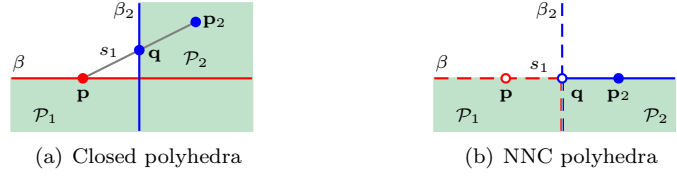
\begin{figure}
\centering
\mbox{
\scriptsize{
\subfigure[Closed polyhedra]{
\label{subfig:lemmas:closed}
\begin{picture}(120,55)(-15,-10)
\put(0,0){
\setlength{\unitlength}{0.7pt}%
\psset{xunit=1cm,yunit=1cm,runit=1cm}
\psset{origin={0,0}}
\pspolygon[linecolor=lightgreen,fillcolor=lightgreen,fillstyle=solid](-0.2,-0.2)(3,-0.2)(3,0.5)(-0.2,0.5)
\pspolygon[linecolor=lightgreen,fillcolor=lightgreen,fillstyle=solid](1.5,-0.2)(3,-0.2)(3,1.5)(1.5,1.5)
\psline[linecolor=red,linestyle=solid](-0.2,0.5)(3,0.5)
\rput(-0.1,0.7){$\beta$}
\psline[linecolor=blue,linestyle=solid](1.5,-0.2)(1.5,1.5)
\rput(1.3,1.4){$\beta_2$}
\psline[linecolor=dgray](0.75,0.5)(2.25,1.25) 
\rput(1.1,0.85){$s_1$}
\pscircle*[linecolor=red](0.75,0.5){2pt}
\rput(0.75,0.27){$\vect{p}$}
\pscircle*[linecolor=blue](2.25,1.25){2pt}
\rput(2.5,1.25){$\vect{p}_2$}
\pscircle*[linecolor=blue](1.5,0.875){2pt}
\rput(1.7,0.8){$\vect{q}$}
\rput(0.25,0){$\cP_1$}
\rput(2.25,0.75){$\cP_2$}
}
\end{picture}
}
\qquad\qquad
\subfigure[NNC polyhedra]{
\label{subfig:lemmas:nnc}
\begin{picture}(120,55)(-15,-10)
\put(0,0){
\setlength{\unitlength}{0.7pt}%
\psset{xunit=1cm,yunit=1cm,runit=1cm}
\psset{origin={0,0}}
\pspolygon[linecolor=lightgreen,fillcolor=lightgreen,fillstyle=solid](-0.2,-0.2)(1.5,-0.2)(1.5,0.5)(-0.2,0.5)
\pspolygon[linecolor=lightgreen,fillcolor=lightgreen,fillstyle=solid](1.5,-0.2)(3,-0.2)(3,0.5)(1.5,0.5)
\psline[linecolor=red,linestyle=dashed,linewidth=0.5pt](1.485,-0.2)(1.485,0.5)
\psline[linecolor=red,linestyle=dashed](-0.2,0.5)(1.5,0.5)
\rput(-0.1,0.7){$\beta$}
\psline[linecolor=blue,linestyle=solid](1.5,0.5)(3,0.5)
\psline[linecolor=blue,linestyle=dashed,linewidth=0.5pt](1.515,-0.2)(1.515,0.5)
\psline[linecolor=blue,linestyle=dashed](1.5,0.5)(1.5,1.5)
\rput(1.3,1.4){$\beta_2$}
\rput(1.2,0.65){$s_1$}
\pscircle*[linecolor=red](0.75,0.5){2pt}
\pscircle*[linecolor=white](0.75,0.5){1.2pt}
\rput(0.75,0.27){$\vect{p}$}
\pscircle*[linecolor=blue](2.25,0.5){2pt}
\rput(2.25,0.27){$\vect{p}_2$}
\pscircle*[linecolor=blue](1.5,0.5){2pt}
\pscircle*[linecolor=white](1.5,0.5){1.2pt}
\rput(1.7,0.27){$\vect{q}$}
\rput(0.25,0){$\cP_1$}
\rput(2.75,0){$\cP_2$}
}
\end{picture}
}
}
}
\caption{Pictorial representations for
Lemmas~\ref{lem:union-of-closed-polyhedra-is-not-convex}
and~\ref{lem:union-of-nnc-polyhedra-is-not-convex}}
\label{fig:lemmas}
\end{figure}

\begin{theorem}
\label{thm:union-of-closed-polyhedra-is-convex}
Let $\cP_1, \cP_2 \in \CPset_n$ be
topologically closed non-empty polyhedra,
where $\cP_1 = \con(\cC_1) = \gen(\cG_1)$.
Then $\cP_1 \uplus \cP_2 \neq \cP_1 \union \cP_2$
if and only if there exist
a constraint $\beta_1 \in \cC_1$ and
a generator $g_1$ in $\cG_1$ such that
\textup{(1)} $g_1$ saturates $\beta_1$,
\textup{(2)} $\cP_2$ violates $\beta_1$, and
\textup{(3)} $\cP_2$ does not subsume $g_1$.
\end{theorem}

\begin{proof}
Suppose first that $\cP_1 \uplus \cP_2 \neq \cP_1 \union \cP_2$.
As  `$\mathord{\uplus}$' is the least upper bound operator for closed polyhedra,
there exist points $\vect{p}_1 \in \cP_1 \setdiff \cP_2$ and
$\vect{p}_2 \in \cP_2 \setdiff \cP_1$ such that
\[
    [\vect{p}_1, \vect{p}_2] \Nsseq (\cP_1 \union \cP_2).
\]
As $\vect{p}_1 \in \cP_1$, there exists a point
\[
  \vect{p}
    \defeq
      (1 - \sigma)\vect{p}_1 + \sigma \vect{p}_2
        \in [\vect{p}_1, \vect{p}_2] \inters \cP_1
\]
such that $\sigma \in \nonnegRset$ is maximal
(note that, by convexity, $\sigma \leq 1$);
then $\vect{p}$ must saturate a constraint $\beta_1 \in \cC_1$.
Moreover $\vect{p} \notin \cP_2$ since, otherwise,
we would have $[\vect{p}_1, \vect{p}] \sseq \cP_1$ and
$[\vect{p}, \vect{p}_2] \sseq \cP_2$, contradicting
$[\vect{p}_1, \vect{p}_2] \Nsseq \cP_1 \union \cP_2$.
Hence $\vect{p}_2$ does not satisfy $\beta_1$
so that $\cP_2$ violates $\beta_1$.
Let $\cG'_1$ be the generator system containing
all the points and rays in $\cG_1$ that saturate $\beta_1$.
Then $\vect{p} \in \gen(\cG'_1)$.
By Theorem~\ref {thm:minkowski-weyl}, as $\vect{p} \notin \cP_2$,
there is a point or ray $g_1$ in $\cG'_1$ that is not subsumed by $\cP_2$.
Hence conditions~(1), (2) and~(3) are all satisfied.

Suppose now that there exist
a constraint $\beta_1 \in \cC_1$ and
a generator $g_1$ in $\cG_1$ such that
conditions~(1), (2) and~(3) hold.
Then, as $\cP_1 = \con(\cC_1)$, $\beta_1$ is satisfied by $\cP_1$.
If $g_1 \defeq \vect{p}_1$ is a point,
then, by letting $\beta \defeq \beta_1$ and $\vect{p} \defeq \vect{p}_1$
in Lemma~\ref{lem:union-of-closed-polyhedra-is-not-convex},
the required three conditions hold so that $\cP_1 \union \cP_2$ is not convex.
Now suppose that $g_1 \defeq \vect{r}_1$ is a ray for $\cP_1$.
Suppose there exists a point $\vect{p}'_1 \in \cP_1$
that saturates the constraint $\beta_1$.
By condition~(3), $\vect{r}_1$ is not a ray for $\cP_2$;
hence for some $\rho \in \nonnegRset$ there exists a point
$\vect{p}_1 \defeq \vect{p}'_1 + \rho \vect{r}_1 \in \cP_1 \setdiff \cP_2$
that also saturates $\beta_1$.
Hence, letting $\beta \defeq \beta_1$ and $\vect{p} \defeq \vect{p}_1$
in Lemma~\ref{lem:union-of-closed-polyhedra-is-not-convex},
the required three conditions hold so that $\cP_1 \union \cP_2$ is not convex.
Otherwise, no point in $\cP_1$ saturates $\beta_1$.%
\footnote{This may happen because we made no minimality assumption
on the constraint system $\cC_1$, so that $\beta_1$ may be redundant.}
Suppose, for some $\vect{a} \in \Rset^n$ and $b \in \Rset$,
$\beta_1 = \bigl( \langle \vect{a}, \vect{x} \rangle \relop b \bigr)$;
then, since $\cP_1 \neq \emptyset$,
there exist a point $\vect{p}'_1 \in \cP_1$ and a constraint
$\beta'_1 \defeq \bigl( \langle \vect{a}, \vect{x} \rangle \relop b' \bigr)$
such that $\cP_1$ satisfies $\beta'_1$ and
$\vect{p}'_1$ saturates $\beta'_1$;
note that $\beta'_1$ is also saturated by ray $\vect{r}_1$.
Thus we can construct, as done above, a point
$\vect{p}_1 \defeq \vect{p}'_1 + \rho \vect{r}_1 \in \cP_1 \setdiff \cP_2$
that saturates $\beta'_1$.
Hence, letting $\beta \defeq \beta'_1$ and $\vect{p} \defeq \vect{p}_1$
in Lemma~\ref{lem:union-of-closed-polyhedra-is-not-convex},
the required three conditions hold so that $\cP_1 \union \cP_2$ is not convex.
Therefore, in all cases,
$\cP_1 \uplus \cP_2 \neq \cP_1 \union \cP_2$.
\qed
\end{proof}

\begin{example}
\label{ex:inexact-closed-polyhedra}
Consider the closed polyhedra
\begin{align*}
  \cP_1 &= \con(\cC_1)
        = \con\bigl(\{ x_1 \geq 0, x_2 \geq 0, x_1 + x_2 \leq 2\}\bigr)\\
        &= \gen(\cG_1)
        = \gen\bigl( (\emptyset, P) \bigr), \\
  \cP_2 &= \con(\cC_2)
        = \con\bigl(\{ x_1 \leq 2, x_2 \geq 0, x_1 - x_2 \geq 0\}\bigr),
\intertext{%
where
\(
  P = \bigl\{
        (0,0)^\transpose, (2,0)^\transpose, (0,2)^\transpose
      \bigr\}
\).
Then
}
  \cP_1 \uplus \cP_2
        &= \con\bigl(\{ x_1 \geq 0, x_2 \geq 0, x_1 \leq 2, x_2 \leq 2\}\bigr)
\end{align*}
so that
\(
  (1,1)^\transpose
    \in (\cP_1 \uplus \cP_2) \setdiff (\cP_1 \union \cP_2)
\)
and, hence, $\cP_1 \uplus \cP_2 \neq \cP_1 \union \cP_2$.
In Theorem~\ref{thm:union-of-closed-polyhedra-is-convex},
let $\beta_1 = (x_1 + x_2 \leq 2)$ and $g_1 = (0,2)^\transpose$.
Then conditions~(1), (2) and~(3) are all satisfied.
\end{example}

For each $i \in \{1, 2\}$, let $l_i$ and $m_i$ denote
the number of constraints in $\cC_i$ and generators in $\cG_i$,
respectively.
Then, the worst-case complexity of an algorithm based on
Theorem~\ref{thm:union-of-closed-polyhedra-is-convex},
computed by summing the complexities for checking
each of the conditions~(1), (2) and~(3),
is in $\bigO\bigl( n (l_1m_1 + l_1m_2 + l_2m_1) \bigr)$.
Notice that the complexity bound is not symmetric so that,
if $l_1 m_1 \gg l_2 m_2$ holds, then an efficiency improvement
can be obtained by exchanging the roles of $\cP_1$ and $\cP_2$
in the theorem.
In all cases, an improvement is obtained with respect to
the $\bigO\bigl( n (l_1 + l_2) m_1 m_2) \bigr)$ complexity bound
of Algorithm~7.1 in~\cite{BemporadFT01}.

\subsection{Exact Join Detection for Not Necessarily Closed Polyhedra}

We now consider the exact join detection problem for two NNC
polyhedra $\cP_1, \cP_2 \in \Pset_n$; in this case,
it can happen that $\cP_1 \uplus \cP_2 \neq \cP_1 \union \cP_2$
although the union $\cP_1 \union \cP_2$ is convex.

\begin{figure}
\centering
\mbox{
\scriptsize{
\subfigure[$\cP, \cQ \in \Pset_2$]{
\label{subfig:nnc-example-P1-P2}
\begin{picture}(100,90)(0,-6)
\put(0,0){
\setlength{\unitlength}{0.7pt}%
\psset{xunit=1cm,yunit=1cm,runit=1cm}
\psset{origin={0,0}}
\psline{->}(-0.5,0)(3,0)
\psline{->}(0,-0.3)(0,2.8)
\rput(-0.2,-0.2){O}
\rput(2.7,-0.2){$x_1$}
\rput(-0.2,2.4){$x_2$}
\pspolygon[linecolor=red,linestyle=dashed,fillcolor=lightgreen,fillstyle=solid](1,0.5)(2,0.5)(2,2)(1,2)
\pscircle*[linecolor=white](1,0.5){2pt}
\pscircle[linecolor=red](1,0.5){2pt}
\rput(0.8,0.3){$A$}
\pscircle*[linecolor=white](2,0.5){2pt}
\pscircle[linecolor=red](2,0.5){2pt}
\rput(2.2,0.3){$B$}
\pscircle*[linecolor=white](2,2){2pt}
\pscircle[linecolor=red](2,2){2pt}
\rput(2.2,2.2){$C$}
\pscircle*[linecolor=white](1,2){2pt}
\pscircle[linecolor=red](1,2){2pt}
\rput(0.8,2.2){$D$}
\pscircle*[linecolor=blue](2,1.25){2pt}
\rput(2.3,1.25){$E$}
\rput(1.55,1.25){$\cP$}
\rput(2.6,1.75){$\cQ$}
\pscurve[linewidth=0.5pt]{->}(2.4,1.75)(2.25,1.6)(2.05,1.32)
}
\end{picture}
}
\
\subfigure[$\cP, \cQ' \in \Pset_2$]{
\label{subfig:nnc-example-P1-P3}
\begin{picture}(100,90)(0,-6)
\put(0,0){
\setlength{\unitlength}{0.7pt}%
\psset{xunit=1cm,yunit=1cm,runit=1cm}
\psset{origin={0,0}}
\psline{->}(-0.5,0)(3,0)
\psline{->}(0,-0.3)(0,2.8)
\rput(-0.2,-0.2){O}
\rput(2.7,-0.2){$x_1$}
\rput(-0.2,2.4){$x_2$}
\pspolygon[linecolor=red,linestyle=dashed,fillcolor=lightgreen,fillstyle=solid](1,0.5)(2,0.5)(2,2)(1,2)
\psline[linecolor=blue,linestyle=solid](2,0.5)(2,2)
\pscircle*[linecolor=white](1,0.5){2pt}
\pscircle[linecolor=red](1,0.5){2pt}
\rput(0.8,0.3){$A$}
\pscircle*[linecolor=white](2,0.5){2pt}
\pscircle[linecolor=red](2,0.5){2pt}
\rput(2.2,0.3){$B$}
\pscircle*[linecolor=white](2,2){2pt}
\pscircle[linecolor=red](2,2){2pt}
\rput(2.2,2.2){$C$}
\pscircle*[linecolor=white](1,2){2pt}
\pscircle[linecolor=red](1,2){2pt}
\rput(0.8,2.2){$D$}
\rput(1.55,1.25){$\cP$}
\rput(2.6,1.75){$\cQ'$}
\pscurve[linewidth=0.5pt]{->}(2.4,1.75)(2.25,1.6)(2.05,1.32)
}
\end{picture}
}
\
\subfigure[$\cP \uplus \cQ = \cP \uplus \cQ'$]{
\label{subfig:nnc-example-poly-hull}
\begin{picture}(100,90)(0,-6)
\put(0,0){
\setlength{\unitlength}{0.7pt}%
\psset{xunit=1cm,yunit=1cm,runit=1cm}
\psset{origin={0,0}}
\psline{->}(-0.5,0)(3,0)
\psline{->}(0,-0.3)(0,2.8)
\rput(-0.2,-0.2){O}
\rput(2.7,-0.2){$x_1$}
\rput(-0.2,2.4){$x_2$}
\pspolygon[linecolor=dgreen,linestyle=dashed,fillcolor=lightgreen,fillstyle=solid](1,0.5)(2,0.5)(2,2)(1,2)
\psline[linecolor=dgreen,linestyle=solid](2,0.5)(2,2)
\pscircle*[linecolor=white](1,0.5){2pt}
\pscircle[linecolor=dgreen](1,0.5){2pt}
\rput(0.8,0.3){$A$}
\pscircle*[linecolor=white](2,0.5){2pt}
\pscircle[linecolor=dgreen](2,0.5){2pt}
\rput(2.2,0.3){$B$}
\pscircle*[linecolor=white](2,2){2pt}
\pscircle[linecolor=dgreen](2,2){2pt}
\rput(2.2,2.2){$C$}
\pscircle*[linecolor=white](1,2){2pt}
\pscircle[linecolor=dgreen](1,2){2pt}
\rput(0.8,2.2){$D$}
}
\end{picture}
}
}
}
\caption{The convex polyhedral hull of NNC polyhedra}
\label{fig:nnc-examples}
\end{figure}
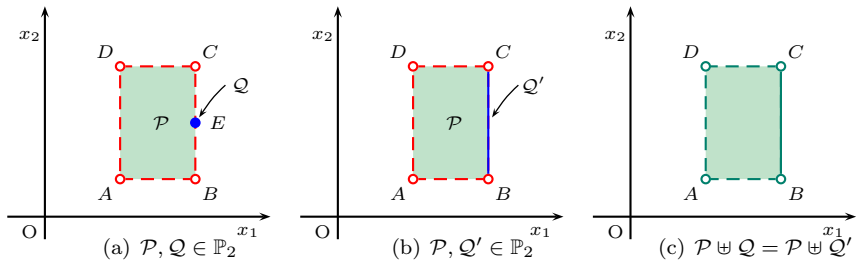

\begin{example}
\label{ex:nnc-polyhedra-convex-union-is-not-the-hull}
Consider the NNC polyhedra $\cP$ and $\cQ$ in
Figure~\ref{subfig:nnc-example-P1-P2}, where $\cP$ is the open
rectangle $ABCD$ and $\cQ$ is the single point $E$.
The union $\cP \union \cQ$ is convex but it is not an NNC polyhedron:
the convex polyhedral hull $\cP \uplus \cQ$
(see Figure~\ref{subfig:nnc-example-poly-hull})
also contains the line segment $(B,C)$ and hence
$\cP \uplus \cQ \neq \cP \union \cQ$.
On the other hand,
if we now consider $\cP$ and $\cQ'$, as shown in
Figure~\ref{subfig:nnc-example-P1-P3},
where $\cQ'$ is the line segment $(B,C)$,
then the convex polyhedral hull $\cP \uplus \cQ'$
is such that $\cP \uplus \cQ' = \cP \uplus \cQ = \cP \union \cQ'$.
\end{example}

Before stating and proving the main result for this section, we
present a lemma similar to
Lemma~\ref{lem:union-of-closed-polyhedra-is-not-convex}
but generalized so as to apply to NNC polyhedra.

\begin{lemma}
\label{lem:union-of-nnc-polyhedra-is-not-convex}
Let $\cP_1, \cP_2 \in \Pset_n$ be non-empty polyhedra.
Suppose that there exist
a constraint $\beta$ and a vector $\vect{p}$ such that
\textup{(1)} $\vect{p}$ saturates $\beta$,
\textup{(2)} $\beta$ is satisfied by $\cP_1$ but violated by $\cP_2$, and
\textup{(3)}  $\vect{p} \in \topclosure(\cP_1) \setdiff \topclosure(\cP_2)$.
Then $\cP_1 \uplus \cP_2 \neq \cP_1 \union \cP_2$.
\end{lemma}

\begin{proof}
(See also Figures~\ref{subfig:lemmas:closed} and~\ref{subfig:lemmas:nnc}.)
By (2), there exists a point $\vect{p}_2 \in \cP_2$ that violates $\beta$.
Consider the line segment $s \defeq (\vect{p}, \vect{p}_2]$;
by (1), no point on $s$ satisfies $\beta$;
by (3), $\vect{p} \in \topclosure(\cP_1)$ so that $s \sseq \cP_1 \uplus \cP_2$.
Also, by (3), $\vect{p} \notin \topclosure(\cP_2)$ so that
there exists a constraint $\beta_2$
that is satisfied by $\topclosure(\cP_2)$ but violated by $\vect{p}$.
Since $\vect{p} \notin \cP_2$ and $\vect{p}_2 \in \cP_2$,
there exists a vector $\vect{q} \in s$ that
saturates $\beta_2$.
It follows that, as $\vect{q} \neq \vect{p}$, the open line segment
$s_1 \defeq (\vect{p}, \vect{q})$ is non-empty
and every point in $s_1$ violates both $\beta$ and $\beta_2$;
hence $s_1 \inters \cP_1 = s_1 \inters \cP_2 = \emptyset$.
However, by construction,
\[
  (\vect{p}, \vect{q})
    \sslt (\vect{p}, \vect{p}_2]
      \sseq (\cP_1 \uplus \cP_2),
\]
so that $\cP_1 \uplus \cP_2 \neq \cP_1 \union \cP_2$.
\qed
\end{proof}

\begin{theorem}
\label{thm:union-of-nnc-polyhedra-iff-hull}
For $i \in \{1,2\}$, let $\cP_i = \con(\cC_i) = \gen(\cG_i) \in \Pset_n$
be non-empty polyhedra.
Then $\cP_1 \uplus \cP_2 \neq \cP_1 \union \cP_2$
if and only if, for some $i,j \in \{1,2\}$, $i \neq j$,
there exists a generator $g_i$ in $\cG_i$
that saturates a constraint $\beta_i \in \cC_i$ violated by $\cP_j$
and at least one of the following hold:
\begin{itemize}
\item[\textup{(1)}]
 $g_i$ is a ray or closure point in $\cG_i$ that is not subsumed by $\cP_j$;
\item[\textup{(2)}]
 $g_i$ is a point in $\cG_i$, $\beta_i$ is non-strict
 and $g_i \notin \topclosure(\cP_j)$;
\item[\textup{(3)}]
 $\beta_i$ is strict and saturated by a point
 $\vect{p} \in (\cP_1 \uplus \cP_2) \setdiff \cP_j$.
\end{itemize}
\end{theorem}

\begin{proof}
Suppose first that $\cP_1 \uplus \cP_2 \neq \cP_1 \union \cP_2$.
As `$\mathord{\uplus}$' is the least upper bound operator for NNC polyhedra,
it follows from the note following Definition~\ref{def:closure-point} that,
for some $i,j \in \{1,2\}$, $i \neq j$,
there exists a closure point $\vect{p}_i$ of  $\cP_i$
and a point $\vect{p}_j \in \cP_j$ such that
\[
  (\vect{p}_i, \vect{p}_j] \Nsseq \cP_1 \union \cP_2.
\]
For ease of notation, we will assume that $i = 1$ and $j = 2$;
the other case follows by a symmetrical argument.
As $\vect{p}_1 \in \topclosure(\cP_1)$, there exists a point
\[
  \vect{p}
    \defeq
      (1 - \sigma)\vect{p}_1 + \sigma \vect{p}_2
        \in [\vect{p}_1, \vect{p}_2] \inters \topclosure(\cP_1)
\]
such that $\sigma \in \nonnegRset$ is maximal
(note that, by convexity, $\sigma < 1$);
then $\vect{p} \in \cP_1 \uplus \cP_2$ and
saturates a constraint $\beta_1 \in \cC_1$
where $\beta_1$ is strict if $\vect{p} \notin \cP_1$.
Note that $\vect{p} \notin \cP_2$ since, otherwise,
we would have $(\vect{p}_1, \vect{p}) \sseq \cP_1$ and
$[\vect{p}, \vect{p}_2] \sseq \cP_2$, contradicting
$(\vect{p}_1, \vect{p}_2] \Nsseq \cP_1 \union \cP_2$.
Moreover, if $\vect{p} \in \cP_1$, $\vect{p} \notin \topclosure(\cP_2)$
since, otherwise, we would have $(\vect{p}_1, \vect{p}] \sseq \cP_1$ and
$(\vect{p}, \vect{p}_2] \sseq \cP_2$, again contradicting
$(\vect{p}_1, \vect{p}_2] \Nsseq \cP_1 \union \cP_2$.

Let $\cG'_1 = (R'_1, P'_1, C'_1)$ be
the system of all the generators in $\cG_1$ that saturate $\beta_1$
so that $\vect{p} \in \gen\bigl((R'_1, P'_1 \union C'_1, \emptyset)\bigr)$.
Suppose condition~(1) does not hold;
that is, suppose that all the rays in $R'_1$ are subsumed by $\cP_2$
and $C'_1 \sseq \topclosure(\cP_2)$.
If $\beta_1$ is non-strict,
$\vect{p} \in \cP_1$ so that $\vect{p} \notin \topclosure(\cP_2)$;
hence, by Theorem~\ref{thm:NNC-minkowski-weyl}, there must exist a generator
point $g_1 \in P'_1 \setdiff \topclosure(\cP_2)$ and condition~(2) holds.
If instead, $\beta_1$ is strict, then, since
$\vect{p} \in \cP_1 \uplus \cP_2$, $\vect{p} \notin \cP_2$
and $\vect{p}$ saturates $\beta_1$, condition~(3) holds.

Suppose now that, for some $i, j \in \{1, 2\}$ $i \neq j$,
there exists a generator $g_i$ in $\cG_i$
that saturates a constraint $\beta_i \in \cC_i$ violated by $\cP_j$
and condition~(1), (2) or (3) holds.
As before, we assume that $i = 1$ and $j = 2$,
since the other case follows by a symmetrical argument.
Let
\(
  \beta_1
    \defeq
      \bigl(
        \langle \vect{a}, \vect{x} \rangle \relop b
      \bigr)
\),
where $\mathord{\relop} \in \{<, \leq\}$.
Suppose condition~(1) holds;
so that $g_1$ is a closure point or ray that is not subsumed by $\cP_2$,
Consider first the case when $g_1$ is a closure point in $\cG_1$
so that $g_1 \notin \topclosure(\cP_2)$.
Then, by letting $\beta \defeq \beta_1$ and $\vect{p} \defeq g_1$
in Lemma~\ref{lem:union-of-nnc-polyhedra-is-not-convex}, it follows that
$\cP_1 \uplus \cP_2 \neq \cP_1 \union \cP_2$.
Consider now the case when $g_1$ is a ray in $\cG_1$.
Since $\cP_1 \neq \emptyset$,
there exist a point $\vect{p}'_1 \in \topclosure(\cP_1)$ and
a constraint
\(
  \beta'_1
    \defeq
      \bigl(
        \langle \vect{a}, \vect{x} \rangle
           \leq \langle \vect{a}, \vect{p}'_1 \rangle
      \bigr)
\)
such that $\cP_1$ satisfies $\beta'_1$;
note that, by definition,
$\beta'_1$ is saturated by the point $\vect{p}'_1$ and the ray $g_1$.%
\footnote{The $\langle \vect{a}, \vect{p}'_1 \rangle$ may differ from
$b$ because we made no minimality assumption on
the constraint system $\cC_1$, so that $\beta_1$ may be redundant.%
}
Therefore, for some $\rho \in \nonnegRset$, the point
\(
  \vect{p}_1
    \defeq
      \vect{p}'_1 + \rho g_1
        \notin \topclosure(\cP_2)
\);
hence, as $\vect{p}_1 \in \topclosure(\cP_1)$ and saturates $\beta'_1$,
by letting $\beta \defeq \beta'_1$ and $\vect{p} \defeq \vect{p}_1$
in Lemma~\ref{lem:union-of-nnc-polyhedra-is-not-convex}, it follows that
$\cP_1 \uplus \cP_2 \neq \cP_1 \union \cP_2$.
If condition~(2) holds,
then $g_1$ is a point in $\cG_1$
(so that $g_1 \in \cP_1$)
and $g_1 \notin \topclosure(\cP_2)$.
Then, by letting $\beta \defeq \beta_1$ and $\vect{p} \defeq g_1$
in Lemma~\ref{lem:union-of-nnc-polyhedra-is-not-convex}, it follows that
$\cP_1 \uplus \cP_2 \neq \cP_1 \union \cP_2$.
Finally suppose that condition~(3) holds.
In this case $\beta_1$ is strict, so that $\vect{p} \notin \cP_1$, and hence
$\vect{p} \in \bigl(\cP_1 \uplus \cP_2\bigr) \setdiff (\cP_1 \union \cP_2)$;
therefore $\cP_1 \uplus \cP_2 \neq \cP_1 \union \cP_2$.
\qed
\end{proof}

Observe that the conditions stated for the NNC case
in Theorem~\ref{thm:union-of-nnc-polyhedra-iff-hull}
are more involved than the conditions stated for the topologically closed
case in Theorem~\ref{thm:union-of-closed-polyhedra-is-convex}.
In particular, a direct correspondence can only be found for
condition~(2) of Theorem~\ref{thm:union-of-nnc-polyhedra-iff-hull}.
The added complexity, which naturally propagates to the corresponding
implementation, is justified by the need to properly capture special
cases where, as said above, convexity alone is not sufficient.

In particular, the check for condition~(3)
in Theorem~\ref{thm:union-of-nnc-polyhedra-iff-hull}
is more expensive than the other checks
and hence should be delayed as much as possible.
Writing $\cH(\beta)$ to denote
the affine hyperplane induced by constraint $\beta$,
\footnote{Namely, if
$\beta = \bigl( \langle \vect{a}, \vect{x} \rangle \relop b \bigr)$,
then
\(
  \cH(\beta)
    = \con\Bigl(
            \bigl\{
              \langle \vect{a}, \vect{x} \rangle = b
            \bigr\}
          \Bigr)
\).
}
condition~(3) can be implemented by checking that
\(
  (\cP_1 \uplus \cP_2) \inters \cH(\beta_i)
    \sseq
      \cP_j \inters \cH(\beta_i)
\)
does not hold.
Even though it is possible to identify cases where optimizations apply,
in the general case the inclusion test above will require the application
of the (incremental) conversion procedure for NNC polyhedra representations.

In the following, we provide a few examples showing cases when
conditions~(1) and~(3)
of Theorem~\ref{thm:union-of-nnc-polyhedra-iff-hull} come into play.

\begin{example}[Condition~(1)
of Theorem~\ref{thm:union-of-nnc-polyhedra-iff-hull}]
We first show how condition~(1) of
Theorem~\ref{thm:union-of-nnc-polyhedra-iff-hull}
where $g_1$ is a closure point can properly discriminate
between the two cases illustrated in Figures~\ref{subfig:nnc-example-P1-P2}
and~\ref{subfig:nnc-example-P1-P3}.

Consider the polyhedra $\cP$ and $\cQ$
in Figure~\ref{subfig:nnc-example-P1-P2} and
assume that the line segment $(B,C)$ satisfies
the constraint $x_1 = 4$.
In the statement of
Theorem~\ref{thm:union-of-nnc-polyhedra-iff-hull},
let $\cP_1 = \cP$, $\cP_2 = \cQ$, $i = 1$, $j = 2$,
$\beta_1 = (x_1 < 4) \in \cC_1$ and $g_1 = B$
be a closure point in $\cG_1$.
Then $\beta_1$ is violated by $\cP_2$ and
saturated by $g_1$, but $g_1$ is not subsumed by $\cP_2$.
Hence condition~(1) of Theorem~\ref{thm:union-of-nnc-polyhedra-iff-hull}
holds and we correctly conclude that $\cP \uplus \cQ \neq \cP \union \cQ$.

On the other hand, if we consider polyhedra $\cP$ and $\cQ'$ in
Figure~\ref{subfig:nnc-example-P1-P3}
and let $\cP_1 = \cP$ and $\cP_2 = \cQ'$,
then the closure point $g_1 = B$ is subsumed by $\cP_2$ so that
condition~(1)
of Theorem~\ref{thm:union-of-nnc-polyhedra-iff-hull} does not hold.

Note that such a discrimination could not be obtained by checking
only condition~(2)
of Theorem~\ref{thm:union-of-nnc-polyhedra-iff-hull}.
If we swap the indices $i$ and $j$ so that $i = 2$, $j = 1$;
letting $\beta_2 = (x_1 \geq 4) \in \cC_2$ and
$g_2 = E$ be a point in $\cG_2$,
then  $g_2 \in \topclosure(\cP)$ and
$\beta_2$ is a non-strict constraint of both $\cQ$ and $\cQ'$
violated by  $\cP$ and saturated by point $g_2$;
hence condition~(2) does not hold for both $\cP_2 = \cQ$ and
for $\cP_2 = \cQ'$.

For an example of an application of
condition~(1) of Theorem~\ref{thm:union-of-nnc-polyhedra-iff-hull}
when $g_1$ is a ray,
consider $\cQ_1$ and $\cQ_2$
in Figure~\ref{subfig:nnc-examples2-P1-P2},
where $\cQ_1 = \con\bigl( \{ 2 \leq x_1 < 4 \} \bigr)$ is an unbounded
strip and $\cQ_2 = \{ A \}$ is a singleton, with $A = (4, 2)^\transpose$.
It can be seen that $\cQ_1 \uplus \cQ_2$, the polyhedron in
Figure~\ref{subfig:nnc-examples2-P1-uplus-P2},
contains the point $B = (4, 0)^\transpose$
which is not a point in $\cQ_1$ or $\cQ_2$,
so that $\cQ_1 \uplus \cQ_2 \neq \cQ_1 \union \cQ_2$.
In the statement of
Theorem~\ref{thm:union-of-nnc-polyhedra-iff-hull},
let $\cP_1 = \cQ_1$, $\cP_2 = \cQ_2$,
$i = 1$, $j = 2$, $\beta_1 = (x_1 < 4) \in \cC_1$ and
$g_1 = (0, 1)^\transpose$ be a ray in $\cG_1$.
Then $\beta_1$ is violated by $\cP_2$ and
saturated by the ray $g_1$;
but $g_1$ is not subsumed by $\cP_2$ so that
condition~(1)
of Theorem~\ref{thm:union-of-nnc-polyhedra-iff-hull} holds.
\end{example}

\begin{figure}
\centering
\mbox{
\scriptsize{
\subfigure[$\cQ_1, \cQ_2 \in \Pset_2$]{
\label{subfig:nnc-examples2-P1-P2}
\begin{picture}(100,90)(0,-8)
\put(0,0){
\setlength{\unitlength}{0.7pt}%
\psset{xunit=1cm,yunit=1cm,runit=1cm}
\psset{origin={0,0}}
\psline{->}(-0.5,0)(3,0)
\psline{->}(0,-0.3)(0,2.8)
\rput(-0.2,-0.2){O}
\rput(2.7,-0.2){$x_1$}
\rput(-0.2,2.4){$x_2$}
\pspolygon[linecolor=lightgreen,fillcolor=lightgreen,fillstyle=solid](1,-0.25)(2,-0.3)(2,2.8)(1,2.8)
\psline[linecolor=red](1,-0.3)(1,2.8)
\psline[linecolor=red,linestyle=dashed](2,-0.3)(2,2.8)
\pscircle*[linecolor=white](2,1){2pt}
\pscircle[linecolor=red](2,1){2pt}
\pscircle*[linecolor=blue](2,1){1pt}
\rput(2.2,1){$A$}
\rput(1.55,1.25){$\cQ_1$}
\rput(2.6,1.55){$\cQ_2$}
\pscurve[linewidth=0.5pt]{->}(2.4,1.55)(2.25,1.38)(2.05,1.07)
}
\end{picture}
}
\
\subfigure[$\cQ_3, \cQ_4 \in \Pset_2$]{
\label{subfig:nnc-examples2-P3-P4}
\begin{picture}(100,90)(0,-8)
\put(0,0){
\setlength{\unitlength}{0.7pt}%
\psset{xunit=1cm,yunit=1cm,runit=1cm}
\psset{origin={0,0}}
\psline{->}(-0.5,0)(3,0)
\psline{->}(0,-0.3)(0,2.8)
\rput(-0.2,-0.2){O}
\rput(2.7,-0.2){$x_1$}
\rput(-0.2,2.4){$x_2$}
\pspolygon[linecolor=lightgreen,fillcolor=lightgreen,fillstyle=solid](0.5,0.5)(2.5,0.5)(2.5,2)(0.5,2)
\psline[linecolor=red,linestyle=dashed,linewidth=0.5pt](0.5,0.5)(1.5,0.5)
\psline[linecolor=red,linestyle=dashed,linewidth=0.5pt](0.5,0.5)(0.5,2)
\psline[linecolor=red,linestyle=dashed,linewidth=0.5pt](0.5,2)(1.5,2)
\psline[linecolor=red,linestyle=dashed,linewidth=0.5pt](1.47,0.5)(1.47,2)
\psline[linecolor=blue,linestyle=dashed,linewidth=0.5pt](1.52,0.5)(2.5,0.5)
\psline[linecolor=blue,linestyle=dashed,linewidth=0.5pt](2.5,0.5)(2.5,2)
\psline[linecolor=blue,linestyle=dashed,linewidth=0.5pt](1.52,2)(2.5,2)
\psline[linecolor=blue,linestyle=dashed,linewidth=0.5pt](1.53,0.5)(1.53,2)
\pscircle*[linecolor=red](0.5,0.5){2pt}
\pscircle*[linecolor=white](0.5,0.5){1.2pt}
\rput(0.3,0.3){$A$}
\pscircle*[linecolor=white](1.5,0.5){2pt}
\pscircle[linecolor=red,linewidth=0.3pt](1.5,0.5){2pt}
\pscircle[linecolor=blue,linewidth=0.3pt](1.5,0.5){1.2pt}
\pscircle*[linecolor=white](1.5,0.5){0.5pt}
\rput(1.5,0.3){$B$}
\pscircle*[linecolor=white](1.5,2){2pt}
\pscircle[linecolor=red,linewidth=0.3pt](1.5,2){2pt}
\pscircle[linecolor=blue,linewidth=0.3pt](1.5,2){1.2pt}
\pscircle*[linecolor=white](1.5,2){0.5pt}
\rput(1.5,2.2){$C$}
\pscircle*[linecolor=red](0.5,2){2pt}
\pscircle*[linecolor=white](0.5,2){1.2pt}
\rput(0.3,2.2){$D$}
\pscircle*[linecolor=blue](2.5,0.5){2pt}
\pscircle*[linecolor=white](2.5,0.5){1.2pt}
\rput(2.7,0.3){$E$}
\pscircle*[linecolor=blue](2.5,2){2pt}
\pscircle*[linecolor=white](2.5,2){1.2pt}
\rput(2.7,2.2){$F$}
\rput(1,1.25){$\cQ_3$}
\rput(2,1.25){$\cQ_4$}
}
\end{picture}
}
\
\subfigure[$\cQ_5, \cQ_6 \in \Pset_2$]{
\label{subfig:nnc-examples2-P5-P6}
\begin{picture}(100,90)(0,-8)
\put(0,0){
\setlength{\unitlength}{0.7pt}%
\psset{xunit=1cm,yunit=1cm,runit=1cm}
\psset{origin={0,0}}
\psline{->}(-0.5,0)(3,0)
\psline{->}(0,-0.3)(0,2.8)
\rput(-0.2,-0.2){O}
\rput(2.7,-0.2){$x_1$}
\rput(-0.2,2.4){$x_2$}
\pspolygon[linecolor=lightgreen,fillcolor=lightgreen,fillstyle=solid](0.5,0.5)(2.5,0.5)(2.5,1.25)(1.5,2)(0.5,1.25)
\psline[linecolor=red,linestyle=solid,linewidth=0.3pt](0.5,0.51)(2,0.51)
\psline[linecolor=red,linestyle=solid](0.5,0.5)(0.5,1.25)
\psline[linecolor=red,linestyle=solid](0.5,1.25)(1.5,2)
\psline[linecolor=red,linestyle=dashed,linewidth=0.5pt](2,0.5)(1.47,2)
\psline[linecolor=blue,linestyle=solid,linewidth=0.3pt](1,0.49)(2.5,0.49)
\psline[linecolor=blue,linestyle=solid](2.5,0.5)(2.5,1.25)
\psline[linecolor=blue,linestyle=solid](1.52,2)(2.5,1.25)
\psline[linecolor=blue,linestyle=dashed,linewidth=0.5pt](1,0.5)(1.53,2)
\pscircle*[linecolor=red](0.5,0.5){2pt}
\rput(0.3,0.3){$A$}
\pscircle*[linecolor=white](2,0.5){2pt}
\pscircle[linecolor=red,linewidth=0.3pt](2,0.5){2pt}
\rput(2,0.3){$B$}
\pscircle*[linecolor=white](1,0.5){2pt}
\pscircle[linecolor=blue,linewidth=0.3pt](1,0.5){2pt}
\rput(1,0.3){$E$}
\pscircle*[linecolor=white](1.5,2){2pt}
\pscircle[linecolor=red,linewidth=0.5pt](1.5,2){2pt}
\pscircle[linecolor=blue,linewidth=0.5pt](1.5,2){1.2pt}
\rput(1.5,2.2){$C$}
\pscircle*[linecolor=red](0.5,1.25){2pt}
\rput(0.3,1.45){$D$}
\pscircle*[linecolor=blue](2.5,0.5){2pt}
\rput(2.7,0.3){$F$}
\pscircle*[linecolor=blue](2.5,1.25){2pt}
\rput(2.7,1.45){$G$}
\rput(1.27,1.25){$\cQ_5$}
\rput(1.85,1.25){$\cQ_6$}
}
\end{picture}
}
}
}
\mbox{
\scriptsize{
\subfigure[$\cQ_1 \uplus \cQ_2$]{
\label{subfig:nnc-examples2-P1-uplus-P2}
\begin{picture}(100,90)(0,-8)
\put(0,0){
\setlength{\unitlength}{0.7pt}%
\psset{xunit=1cm,yunit=1cm,runit=1cm}
\psset{origin={0,0}}
\psline{->}(-0.5,0)(3,0)
\psline{->}(0,-0.3)(0,2.8)
\rput(-0.2,-0.2){O}
\rput(2.7,-0.2){$x_1$}
\rput(-0.2,2.4){$x_2$}
\pspolygon[linecolor=lightgreen,fillcolor=lightgreen,fillstyle=solid](1,-0.25)(2,-0.3)(2,2.8)(1,2.8)
\psline[linecolor=dgreen](1,-0.3)(1,2.8)
\psline[linecolor=dgreen](2,-0.3)(2,2.8)
\pscircle*[linecolor=dgreen](2,1){2pt}
\rput(2.2,1){$A$}
\pscircle*[linecolor=dgreen](2,0){2pt}
\rput(2.2,0.2){$B$}
\rput(1.5,1.25){$\cQ_1 \uplus \cQ_2$}
}
\end{picture}
}
\
\subfigure[$\cQ_3 \uplus \cQ_4$]{
\label{subfig:nnc-examples2-P3-uplus-P4}
\begin{picture}(100,90)(0,-8)
\put(0,0){
\setlength{\unitlength}{0.7pt}%
\psset{xunit=1cm,yunit=1cm,runit=1cm}
\psset{origin={0,0}}
\psline{->}(-0.5,0)(3,0)
\psline{->}(0,-0.3)(0,2.8)
\rput(-0.2,-0.2){O}
\rput(2.7,-0.2){$x_1$}
\rput(-0.2,2.4){$x_2$}
\pspolygon[linecolor=lightgreen,fillcolor=lightgreen,fillstyle=solid](0.5,0.5)(2.5,0.5)(2.5,2)(0.5,2)
\psline[linecolor=dgreen,linestyle=dashed,linewidth=0.5pt](0.5,0.5)(2.5,0.5)
\psline[linecolor=dgreen,linestyle=dashed,linewidth=0.5pt](0.5,0.5)(0.5,2)
\psline[linecolor=dgreen,linestyle=dashed,linewidth=0.5pt](0.5,2)(2.5,2)
\psline[linecolor=dgreen,linestyle=dashed,linewidth=0.5pt](2.5,0.5)(2.5,2)
\psline[linecolor=dgreen,linestyle=solid](1.5,0.5)(1.5,2)
\pscircle[linecolor=dgreen](0.5,0.5){2pt}
\pscircle[linecolor=white](0.5,0.5){1.2pt}
\rput(0.3,0.3){$A$}
\pscircle*[linecolor=dgreen](1.5,0.5){2pt}
\pscircle*[linecolor=white](1.5,0.5){1.2pt}
\rput(1.5,0.3){$B$}
\pscircle*[linecolor=dgreen](1.5,2){2pt}
\pscircle*[linecolor=white](1.5,2){1.2pt}
\rput(1.5,2.2){$C$}
\pscircle*[linecolor=dgreen](0.5,2){2pt}
\pscircle*[linecolor=white](0.5,2){1.2pt}
\rput(0.3,2.2){$D$}
\pscircle*[linecolor=dgreen](2.5,0.5){2pt}
\pscircle*[linecolor=white](2.5,0.5){1.2pt}
\rput(2.7,0.3){$E$}
\pscircle*[linecolor=dgreen](2.5,2){2pt}
\pscircle*[linecolor=white](2.5,2){1.2pt}
\rput(2.7,2.2){$F$}
\pscircle*[linecolor=dgreen](1.5,0.8){2pt}
\rput(1.75,0.8){$G$}
\rput(1.455,1.25){$\cQ_3 \uplus \cQ_4$}
}
\end{picture}
}
\
\subfigure[$\cQ_5 \uplus \cQ_6$]{
\label{subfig:nnc-examples2-P5-uplus-P6}
\begin{picture}(100,90)(0,-8)
\put(0,0){
\setlength{\unitlength}{0.7pt}%
\psset{xunit=1cm,yunit=1cm,runit=1cm}
\psset{origin={0,0}}
\psline{->}(-0.5,0)(3,0)
\psline{->}(0,-0.3)(0,2.8)
\rput(-0.2,-0.2){O}
\rput(2.7,-0.2){$x_1$}
\rput(-0.2,2.4){$x_2$}
\pspolygon[linecolor=lightgreen,fillcolor=lightgreen,fillstyle=solid](0.5,0.5)(2.5,0.5)(2.5,1.25)(1.5,2)(0.5,1.25)
\psline[linecolor=dgreen,linestyle=solid](0.5,0.5)(2.5,0.5)
\psline[linecolor=dgreen,linestyle=solid](0.5,0.5)(0.5,1.25)
\psline[linecolor=dgreen,linestyle=solid](0.5,1.25)(1.5,2)
\psline[linecolor=dgreen,linestyle=solid](2.5,0.5)(2.5,1.25)
\psline[linecolor=dgreen,linestyle=solid](1.52,2)(2.5,1.25)
\psline[linecolor=dgreen,linestyle=dashed,linewidth=1pt](0.2,2)(2.8,2)
\pscircle*[linecolor=dgreen](0.5,0.5){2pt}
\rput(0.3,0.3){$A$}
\pscircle*[linecolor=white](1.5,2){2pt}
\pscircle[linecolor=dgreen,linewidth=0.5pt](1.5,2){2pt}
\rput(1.5,2.2){$C$}
\pscircle*[linecolor=dgreen](0.5,1.25){2pt}
\rput(0.3,1.45){$D$}
\pscircle*[linecolor=dgreen](2.5,0.5){2pt}
\rput(2.7,0.3){$F$}
\pscircle*[linecolor=dgreen](2.5,1.25){2pt}
\rput(2.7,1.45){$G$}
\rput(1.48,1.25){$\cQ_5 \uplus \cQ_6$}
}
\end{picture}
}
}
}
\caption{More examples for the convex polyhedral hull of NNC polyhedra}
\label{fig:nnc-examples2}
\end{figure}
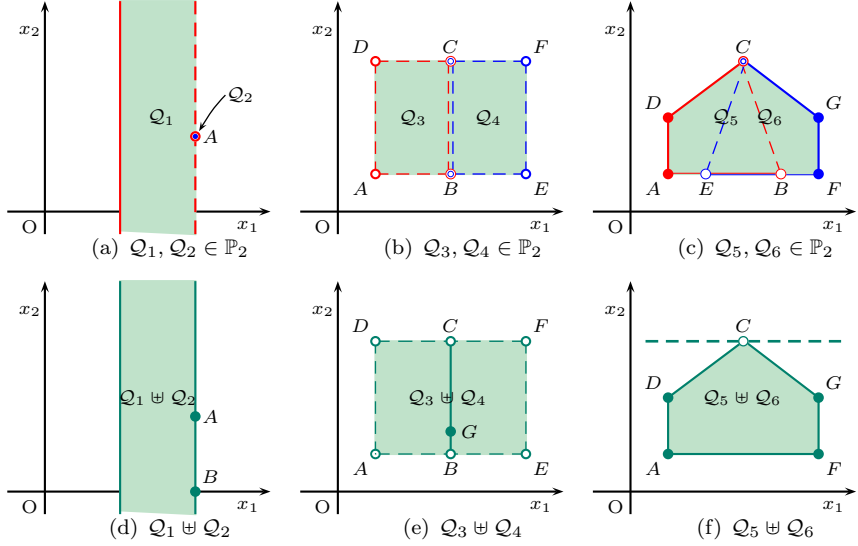

\begin{example}[Condition~(3)
of Theorem~\ref{thm:union-of-nnc-polyhedra-iff-hull}]
This example shows how condition~(3) of
Theorem~\ref{thm:union-of-nnc-polyhedra-iff-hull}
can properly discriminate
between the two cases illustrated in
Figures~\ref{subfig:nnc-examples2-P3-P4}
and~\ref{subfig:nnc-examples2-P5-P6}.

Consider the polyhedra $\cQ_3$ and $\cQ_4$ in
Figure~\ref{subfig:nnc-examples2-P3-P4},
where $\cQ_3$ is the open rectangle $ABCD$,
with the open bound $(B,C)$ defined by the strict constraint $x_1 < 3$,
whereas $\cQ_4$ is the open rectangle $BEFC$.
Then $B = (3,1)^\transpose$ and $C = (3,5)^\transpose$
are closure points for both $\cQ_3$ and $\cQ_4$.
It can be seen that $\cQ_3 \uplus \cQ_4$, the polyhedron in
Figure~\ref{subfig:nnc-examples2-P3-uplus-P4},
contains the open line segment $(B,C)$
so that $\cQ_3 \uplus \cQ_4 \neq \cQ_3 \union \cQ_4$.
In the statement of
Theorem~\ref{thm:union-of-nnc-polyhedra-iff-hull},
let $\cP_1 = \cQ_3$, $\cP_2 = \cQ_4$, $i = 1$, $j = 2$,
$\beta_1 = (x_1 < 3) \in \cC_1$ and $g_1 = B$
be a closure point in $\cG_1$.
Then $\beta_1$ is violated by $\cP_2$ and
saturated by the closure point $g_1$.
Although condition~(1)
does not hold because $g_1$ is subsumed by $\cP_2$,
condition~(3) does hold
since $\beta_1$ is strict and,
taking $\vect{p} = G \in (B, C)$,
we have $\vect{p} \in (\cP_1 \uplus \cP_2) \setdiff \cP_2$.

It is worth stressing that none of the (closure) points in
the open segment $(B,C)$ belong to the generator systems
of $\cP_1$ and $\cP_2$.
The reader is also warned that, even though in this particular example
$\cP_1$, $\cP_2$ and the segment $(B,C)$ are pairwise disjoint
(which trivially implies that the join $\cP_1 \uplus \cP_2$ is inexact),
such a property would not generalize to higher dimensional vector spaces
and hence it cannot be used as a replacement for condition~(3)
in Theorem~\ref{thm:union-of-nnc-polyhedra-iff-hull}.

Consider the polyhedra $\cQ_5$ and $\cQ_6$ in
Figure~\ref{subfig:nnc-examples2-P5-P6},
where $\cQ_5$ is the quadrilateral $ABCD$
and $\cQ_6$ is the quadrilateral $EFGC$.
Then the convex polyhedral hull $\cQ_5 \uplus \cQ_6$ shown in
Figure~\ref{subfig:nnc-examples2-P5-uplus-P6}
is equal to their union $\cQ_5 \union \cQ_6$.
In the statement of
Theorem~\ref{thm:union-of-nnc-polyhedra-iff-hull},
let $\cP_1 = \cQ_5$, $\cP_2 = \cQ_6$, $i = 1$,
$j = 2$, $\beta_1 \in \cC_1$ be the strict constraint
defining the dashed line boundary $(B,C)$ and
$g_1$ be the closure point $C$ in both $\cP_1$ and $\cP_2$.
Then none of the conditions in
Theorem~\ref{thm:union-of-nnc-polyhedra-iff-hull} hold.
\end{example}

\section{Exact Join Detection for Boxes
               and Other Cartesian Products}
\label{sec:boxes}

A rational interval constraint for a dimension $i \in \{ 1, \ldots, n \}$
has the form $x_i \relop b$,
where $\mathord{\relop} \in \{ <, \leq, =, \geq, > \}$
and $b \in \Qset$.
A finite system of rational interval constraints
defines an NNC polyhedron in $\Pset_n$ that we call a \emph{rational box};
the set of all rational boxes in the $n$-dimensional vector space
is denoted $\NNCBox_n$ and is a meet-sublattice of $\Pset_n$.
The domain $\NNCBox_n$ so defined can be seen as the Cartesian product
of $n$ possibly infinite intervals with rational, possibly open boundaries.
If we denote by $\NNCInterval$ the set of such intervals and by
`$\mathord{\oplus}$' the binary join operator over
the bounded join-semilattice
$(\NNCInterval, \mathord{\sseq})$, we have, for each $B_1, B_2 \in \NNCBox$,
\[
  B_1 \uplus B_2
    =
      \bigl(\proj_1(B_1) \oplus \proj_1(B_2)\bigr)
      \times \cdots \times
      \bigl(\proj_n(B_1) \oplus \proj_n(B_2)\bigr).
\]

The following theorem defines a necessary and sufficient condition
that is only based on `$\mathord{\oplus}$' and on the subset ordering
over $\NNCInterval$.  Notice, in particular, that convexity does
not play any role, neither in the statement, nor in the proof.
\begin{theorem}
\label{thm:union-of-boxes-iff-hull}
Let $B_1, B_2 \in \NNCBox_n$.
Then $B_1 \uplus B_2 \neq B_1 \union B_2$ if and only if
\begin{enumerate}
\item
\(
  \exists i \in \{ 1, \ldots, n \}
    \st
      \proj_i(B_1) \oplus \proj_i(B_2)
        \neq
          \proj_i(B_1) \union \proj_i(B_2)
\); or
\item
\(
  \exists i,j \in \{ 1, \ldots, n \}
    \st
      i \neq j
        \land
      \proj_i(B_1) \Nsseq \proj_i(B_2)
        \land
      \proj_j(B_2) \Nsseq \proj_j(B_1)
\).
\end{enumerate}
\end{theorem}
\begin{proof}
Suppose that $B_1 = \emptyset$ so that,
for each $i \in \{ 1, \ldots, n \}$, $\proj_i(B_1) = \emptyset$.
Then, neither condition~(1) nor condition~(2) can hold,
so that the lemma holds.
By a symmetric reasoning, the lemma holds if $B_2 = \emptyset$.
Hence, in the following
we assume that both $B_1$ and $B_2$ are non-empty boxes.

Suppose first that $B_1 \uplus B_2 \neq B_1 \union B_2$;
then there exists a point $\vect{p} \in B_1 \uplus B_2$
such that $\vect{p} \notin B_1$ and $\vect{p} \notin B_2$.
Hence, for some $i,j \in \{ 1, \ldots, n\}$, we have that
$p_i \notin \proj_i(B_1)$ and
$p_j \notin \proj_j(B_2)$.
Note that as $\vect{p} \in B_1 \uplus B_2$, we also have
$p_i \in \proj_i(B_1) \oplus \proj_i(B_2)$ and
$p_j \in \proj_j(B_1) \oplus \proj_j(B_2)$.
Suppose that condition~(1) does not hold.
Then $p_i \in \proj_i(B_2)$ and
$p_j \in \proj_j(B_1)$;
hence we must have $i \neq j$
and $p_i \in \proj_i(B_1) \setdiff \proj_i(B_2)$ and
$p_j \in \proj_j(B_2) \setdiff \proj_j(B_1)$;
implying that
$\proj_i(B_1) \Nsseq \proj_i(B_2)$ and
$\proj_j(B_2) \Nsseq \proj_j(B_1)$,
so that condition~(2) holds.

Assuming that condition~(1) or~(2) holds,
we now prove $B_1 \uplus B_2 \neq B_1 \union B_2$.
First, suppose that condition~(1) holds.
Then there exists
$v \in \proj_i(B_1 \uplus B_2)$
such that
$v \notin \proj_i(B_1)$ and $v \notin \proj_i(B_2)$.
By definition of $\proj_i$,
there exist a point
$\vect{p} \in B_1 \uplus B_2$ such that
$\proj_i(\vect{p}) = v$, so that
$\vect{p} \notin B_1$ and $\vect{p} \notin B_2$;
therefore $B_1 \uplus B_2 \neq B_1 \union B_2$.
Secondly, suppose that condition~(2) holds.
Then there exist values
$v_i \in \proj_i(B_1) \setdiff \proj_i(B_2)$ and
$v_j \in \proj_j(B_2) \setdiff \proj_j(B_1)$;
hence, there exist points
$\vect{p}_i \in B_1$ and $\vect{p}_j \in B_2$ such that
$\proj_i(\vect{p}_i) = v_i$ and $\proj_j(\vect{p}_j) = v_j$.
Let $\vect{p}$ be such that
$\proj_k(\vect{p}) = \proj_k(\vect{p}_i)$,
for all $k \in \{ 1, \ldots, n \} \setdiff \{j\}$,
and $\proj_j(\vect{p}) = v_j$;
then $\vect{p} \notin B_1 \union B_2$.
By definition of the `$\mathord{\uplus}$' operator,
$\vect{p} \in B_1 \uplus B_2$, so that
$B_1 \uplus B_2 \neq B_1 \union B_2$.
\qed
\end{proof}

\begin{example}
\label{ex:box-examples}
Consider the topologically closed boxes
\begin{align*}
  B_1 &= \con\bigl(
               \{ 0 \leq x_1 \leq 1, 0 \leq x_2 \leq 2 \}
             \bigr), \\
  B_2 &= \con\bigl(
               \{ 3 \leq x_1 \leq 4, 0 \leq x_2 \leq 2 \}
             \bigr), \\
  B_3 &= \con\bigl(
               \{ 0 \leq x_1 \leq 4, 1 \leq x_2 \leq 2 \}
             \bigr).
\end{align*}
Then we obtain
\[
  B_1 \uplus B_2
    = B_1 \uplus B_3
      = \con\bigl(
              \{ 0 \leq x_1 \leq 4, 0 \leq x_2 \leq 2 \}
            \bigr).
\]
Letting $\vect{p} = (2,0)^\transpose$,
we have $\vect{p} \in B_1 \uplus B_2$
although $\vect{p} \notin B_1 \union B_2 \union B_3$;
hence $B_1 \uplus B_2 \neq B_1 \union B_2$ and
$B_1 \uplus B_3 \neq B_1 \union B_3$,
i.e., both join computations are inexact.
Observe that
\[
  \proj_1(B_1) \oplus \proj_1(B_2) \neq \proj_1(B_1) \union \proj_1(B_2),
\]
so that, for boxes $B_1$ and $B_2$, condition~(1) holds;
on the other hand we have
\[
  \proj_1(B_3) \Nsseq \proj_1(B_1)
  \quad\text{and}\quad
  \proj_2(B_1) \Nsseq \proj_2(B_3),
\]
so that, for boxes $B_1$ and $B_3$, condition~(2) holds.
\end{example}

This result has been introduced for rational boxes for simplicity only.
Indeed, it trivially generalizes to any Cartesian product of 1-dimensional
numerical abstractions, including:
the well-known abstract domain of multi-dimensional, integer-valued intervals
\cite{CousotC76};
1-dimensional congruence equations like $x = 0 \pmod 2$;
\emph{modulo intervals} \cite{NakanishiF01,NakanishiJPF99}; and
\emph{circular linear progressions} \cite{SenS07}.
For full generality, for each $i \in \{ 1, \ldots, n \}$,
let $(\Aiset{i}, \sseq)$, with $\emptyset \in \Aiset{i} \sseq \wp(\Rset)$,
be a bounded join-semilattice where the binary join operator
is denoted by `$\mathord{\oplus_i}$'.
$(\Aiset{i}, \mathord{\sseq})$ is thus an abstract domain suitable for
approximating $\wp(\Rset)$.
Then, the trivial combination of the $n$ domains $\Aiset{i}$ by means
of Cartesian product,
\(
  \Aset_n \defeq \Aiset{1} \times \dots \times \Aiset{n}
\),
is an abstract domain suitable for approximating
$\wp(\Rset^n)$.\footnote{This construction is called a \emph{direct product}
in the field of abstract interpretation.  The resulting domain is
said to be \emph{attribute-independent}, in the sense that
relational information is not captured.  In other words,
the constraints on space dimension
$i$ are unrelated to those on space dimension $j$ whenever $i \neq j$.}
Theorem~\ref{thm:union-of-boxes-iff-hull} immediately generalizes
to any domain $\Aset_n$ so obtained.

An algorithm for the exact join detection on $\Aset_n$
based on Theorem~\ref{thm:union-of-boxes-iff-hull} will compute,
in the worst case, a linear number of 1-dimensional joins
(applying the `$\mathord{\oplus_i}$' operators)
and a linear number of 1-dimensional inclusion tests.
Since these 1-dimensional operations take constant time,
the worst-case complexity bound for $n$-dimensional boxes is $\bigO(n)$.

\section{Exact Join Detection for Bounded Difference Shapes}
\label{sec:bd-shapes}

A (rational) bounded difference is a non-strict inequality constraint
having one of the forms $\pm x_i \leq b$ or $x_i - x_j \leq b$,
where $i,j \in \{ 1, \ldots, n \}$, $i \neq j$ and $b \in \Qset$.
A finite system of bounded differences defines
a \emph{bounded difference shape} (BD shape);
the set of all BD shapes in the $n$-dimensional vector space
is denoted $\BDshapes_n$ and it is a meet-sublattice of $\CPset_n$.
In this section we specialize the result on topologically closed polyhedra
to the case of BD shapes,
which can be efficiently represented and manipulated as weighted graphs.

\subsection{BD Shapes and their Graph Representation}

We first introduce some notation and terminology
(see also~\cite{BagnaraHMZ05,BagnaraHZ09FMSD,Mine01a,Mine05th}).

Let $\Qset_\infty \defeq \Qset \union \{ +\infty \}$ be totally
ordered by the extension of `$\mathord{<}$' such that $d < +\infty$ for
each $d \in \Qset$.
Let $\cN$ be a finite set of \emph{nodes}.
A \emph{weighted directed graph} (graph, for short)
$G$ in $\cN$ is a pair $(\cN, w)$,
where $\fund{w}{\cN \times \cN}{\Qset_\infty}$ is the
weight function for $G$.
A pair $(n_i,n_j) \in \cN \times \cN$ is an \emph{arc} of $G$
if $w(n_i,n_j) < +\infty$;
the arc is \emph{proper} if $n_i \neq n_j$.
A \emph{path} $\path = n_0 \cdots n_p$ in a graph $G = (\cN, w)$
is a non-empty and finite sequence of nodes such that,
for all $i \in \{1, \ldots, p\}$,
$(n_{i-1}, n_i)$ is an arc of $G$;
each arc $(n_{i-1}, n_i)$ is said to be \emph{in} the path $\path$.
If $\path_1 = n_0 \cdots n_h$
and $\path_2 = n_h \cdots n_p$ are paths in $G$,
where $0 \leq h \leq p$,
then the path concatenation $\path = n_0 \cdots n_h \cdots n_p$ of
$\path_1$ and $\path_2$ is denoted by $\path_1 \pathconc \path_2$;
if $\path_1 = n_0 n_1$ (so that $h = 1$),
then $\path_1 \pathconc \path_2$ will also be denoted by $n_0 \cdot \path_2$.
Note that path concatenation is not the same as sequence concatenation.
The path $\path$ is \emph{simple} if each node
occurs at most once in $\path$;
it is \emph{proper} if all the arcs in it are proper;
it is a \emph{proper cycle} if it is a proper path
and $n_0 = n_p$ (so that $p \geq 2$).
The path $\path$ has \emph{weight}
$w(\path) \defeq \sum_{i=1}^p w(n_{i-1}, n_i)$.
A graph is \emph{consistent} if it has no negative weight cycles.
The set $\Graphs$ of consistent graphs in $\cN$ is partially ordered
by the relation `$\mathord{\graphleq}$' defined,
for all $G_1 = (\cN, w_1)$ and $G_2 = (\cN, w_2)$, by
\[
  G_1 \graphleq G_2
    \quad\iff\quad
      \forall i, j \in \cN \itc w_1(i,j) \leq w_2(i,j).
\]
When augmented with a bottom element $\bot$ representing inconsistency,
this partially ordered set becomes a (non-complete) lattice
\(
  \Graphs_\bot
    = \bigl\langle
        \Graphs \union \{ \bot \},
        \graphleq,
        \graphglb,
        \graphlub
      \bigr\rangle
\),
where `$\mathord{\graphglb}$' and `$\mathord{\graphlub}$'
denote the (finitary) greatest lower bound and least upper bound
operators, respectively.

\begin{definition}
\label{def:closure}
\summary{(Graph closure/reduction.)}
A consistent graph $G = (\cN, w)$ is \emph{(shortest-path) closed}
if the following properties hold:
\begin{align}
\label{rule:zero-weight-self-loops}
  \forall i \in \cN
    &\itc w(i,i) = 0; \\
\label{rule:transitivity}
  \forall i,j,k \in \cN
    &\itc w(i,j) \leq w(i,k) + w(k,j).
\end{align}
The \emph{closure} of a consistent graph $G$ in $\cN$ is
\[
  \closure(G)
    \defeq
      \biggraphlub
        \bigl\{\,
          G^\rc \in \Graphs
        \bigm|
          \text{$G^\rc \graphleq G$ and $G^\rc$ is closed\/}
        \,\bigr\}.
\]
A consistent graph $R$ in $\cN$ is \emph{(shortest-path) reduced} if,
for each graph $G \neq R$ such that $R \graphleq G$,
$\closure(R) \neq \closure(G)$.
A \emph{reduction} for the consistent graph $G$ is
any reduced graph $R$ such that $\closure(R) = \closure(G)$.
\end{definition}
Note that a reduction $R$ for a closed graph $G$
is a \emph{subgraph} of $G$, meaning that all the arcs in $R$
are also arcs in $G$ and have the same finite weight.

Any system of bounded differences in $n$ dimensions defining a
non-empty element $\bd \in \BDshapes_n$ can be represented
by a consistent graph $G = (\cN, w)$
where $\cN = \{ 0, \ldots, n \}$ is the set of graph nodes;
each node $i > 0$ corresponds to the space dimension $x_i$
of the vector space, while $0$ (the \emph{special node})
represents a further space dimension whose value is fixed to zero.
Each arc $(i, j)$ of $G$ denotes the bounded difference
$x_i - x_j \leq w(i, j)$ if $i, j > 0$,
$x_i \leq w(i, 0)$ if $j = 0$ and
$-x_j \leq w(0, j)$ if $i = 0$.
Conversely, it can be seen that, by inverting the above mapping,
each consistent graph $G = (\cN, w)$
where $\cN = \{ 0, \ldots, n \}$
represents a non-empty element $\bd \in \BDshapes_n$.
Graph closure provides a normal form for non-empty BD shapes.
Informally, a closed (resp., reduced) graph encodes
a system of bounded difference constraints
which is closed by entailment (resp., contains no redundant constraint).

If the non-empty BD shapes $\bd_1, \bd_2 \in \BDshapes_n$
are represented by closed graphs
$G_1 = (\cN, w_1)$ and $G_2 = (\cN, w_2)$, respectively,
then the BD shape join $\bd_1 \uplus \bd_2$
is represented by the graph least upper bound
$G_1 \graphlub G_2 = (\cN, w)$,
where $w(i,j) \defeq \max\bigl(w_1(i,j), w_2(i,j))\bigr)$
for each $i,j \in \cN$;
$G_1 \graphlub G_2$ is also closed.
Observe too that the set intersection $\bd_1 \inters \bd_2$
is represented by the graph greatest lower bound $G_1 \graphglb G_2$.

\subsection{Exact Join Detection for Rational BD Shapes}

The following result can be used as the specification of an
exact join decision procedure specialized for rational BD shapes.

\begin{theorem}
\label{thm:union-of-bds-iff-hull}
For each $h \in \{ 1, 2 \}$,
let $\bd_h \in \BDshapes_n$ be a non-empty BD shape
represented by the closed graph $G_h = (\cN, w_h)$ and
let $R_h$ be a subgraph of $G_h$ such that $\closure(R_h) = G_h$.
Let also $G_1 \graphlub G_2 = (\cN, w)$.
Then $\bd_1 \uplus \bd_2 \neq \bd_1 \union \bd_2$ if and only if
there exist arcs $(i, j)$ of $R_1$ and $(k, \ell)$ of $R_2$
such that
\begin{itemize}
\item[\textup{(1)}]
$w_1(i, j) < w_2(i, j)$ and
$w_2(k, \ell) < w_1(k, \ell)$;
and
\item[\textup{(2)}]
$w_1(i, j) + w_2(k, \ell) < w(i, \ell) + w(k, j)$.
\end{itemize}
\end{theorem}
\begin{proof}
Suppose that $\bd_1 \uplus \bd_2 \neq \bd_1 \union \bd_2$,
so that there exists $\vect{p} \in \bd_1 \uplus \bd_2$
such that $\vect{p} \notin \bd_1$ and $\vect{p} \notin \bd_2$.
Hence, there exist $i,j,k,\ell \in \cN$ such that
$(i,j)$ is an arc of $R_1$ satisfying%
\footnote{We extend notation by letting $\proj_0(\vect{v}) \defeq 0$,
for each vector $\vect{v} = (v_1, \ldots, v_n)^\transpose$.}
$\proj_i(\vect{p}) - \proj_j(\vect{p}) > w_1(i, j)$
and $(k, \ell)$ is an arc of $R_2$ satisfying
$\proj_k(\vect{p}) - \proj_\ell(\vect{p}) > w_2(k, \ell)$.
However, as $\vect{p} \in \bd_1 \uplus \bd_2$,
$\proj_i(\vect{p}) - \proj_j(\vect{p}) \leq w(i,j)$ and
$\proj_k(\vect{p}) - \proj_\ell(\vect{p}) \leq w(k, \ell)$
so that, by definition of $G_1 \graphlub G_2$,
we have $w_1(i, j) < w_2(i, j)$ and $w_2(k, \ell) < w_1(k, \ell)$;
hence condition (1) holds.
Since $\vect{p} \in \bd_1 \uplus \bd_2$,
\begin{align*}
  w(i, \ell) + w(k, j)
    &\geq \proj_i(\vect{p}) - \proj_\ell(\vect{p})
      + \proj_k(\vect{p}) - \proj_j(\vect{p}) \\
    &= \proj_i(\vect{p}) - \proj_j(\vect{p})
         + \proj_k(\vect{p}) - \proj_\ell(\vect{p}) \\
    &> w_1(i, j) + w_2(k, \ell).
\end{align*}
Therefore, condition (2) also holds.

We now suppose that there exist
arcs $(i,j)$ of $R_1$ and $(k, \ell)$ of $R_2$
such that conditions (1) and (2) hold.
As $G_1$ and $G_2$ are closed, $w_1(i, i) = w_2(i, i) = 0$ and
$w_1(k, k) = w_2(k, k) = 0$
so that condition (1) implies $i \neq j$ and $k \neq \ell$.
As $G_1 \graphlub G_2$ is closed, $w(i, i) = w(k, k) = 0$
so that, if $i = \ell$ and $j = k$ both hold,
condition (2) implies $w_1(i, j) + w_2(j, i) < 0$;
hence, the graph greatest lower bound $G_1 \graphglb G_2$
contains the negative weight proper cycle $i \cdot j \cdot i$
and thus is inconsistent;
hence $\bd_1 \inters \bd_2 = \emptyset$;
and hence $\bd_1 \uplus \bd_2 \neq \bd_1 \union \bd_2$.
Therefore, in the following we assume that $i \neq \ell$ or $j \neq k$ hold.
If the right hand side of the inequalities in
conditions (1) and (2) are all unbounded,
let $\epsilon \defeq 1$;
otherwise let
\begin{equation*}
  \epsilon
    \defeq
      \min\left\{
        \begin{aligned}
          &w(i, j) - w_1(i, j), \\
          &w(k, \ell) - w_2(k, \ell), \\
          &\frac{1}{2}
             \bigl(
               w(i, \ell) + w(k, j) - w_1(i, j) - w_2(k, \ell)
             \bigr)
        \end{aligned}
      \right\}.
\end{equation*}
Then, by conditions (1) and (2), $\epsilon > 0$.
Consider the graph $G' = (\cN, w')$ where,
for each $r, s \in \cN$,
\begin{equation*}
  w'(r,s)
    \defeq
      \begin{cases}
        - w_1(i, j) - \epsilon,
          &\text{if $(r, s) = (j, i)$;} \\
        - w_2(k, \ell) - \epsilon,
          &\text{if $(r, s) = (\ell, k)$;} \\
        w(r, s),
          &\text{otherwise.}
      \end{cases}
\end{equation*}
We show that $G'$ is a consistent graph;
to this end, since $G \defeq G_1 \graphlub G_2$ is known to be consistent,
it is sufficient to consider the proper cycles of $G'$ that contain
arcs $(j, i)$ or $(\ell, k)$.
Let $\path_{ij} = i\cdots j$ and $\path_{k\ell} = k\cdots \ell$ be
arbitrary simple paths from $i$ to $j$ and from $k$ to $\ell$, respectively.
Then $G'$ is consistent if and only if
$w'(\path_{ij} \cdot i) \geq 0$ and
$w'(\path_{k\ell} \cdot k) \geq 0$.
We only prove $w'(\path_{ij} \cdot i) \geq 0$
since the proof that
$w'(\path_{k\ell} \cdot k) \geq 0$ follows by a symmetrical argument.
As $\path_{ij}$ is simple, it does not contain the arc $(j, i)$.
Suppose first that $\path_{ij}$ does not contain the arc $(\ell, k)$.
Then
\begin{align*}
  w'(\path_{ij}\cdot i)
    &= w'(\path_{ij}) + w'(j, i) \\
    &= w(\path_{ij}) - w_1(i, j) - \epsilon
&\just{def.\ of $w'$} \\
    &\geq w(i, j) - w_1(i, j) - \epsilon
&\just{$G$ closed} \\
    &\geq 0
&\just{def.\ of $\epsilon$}.
\end{align*}
Suppose now that
$\path_{ij} = \path_{i\ell} \pathconc (\ell, k) \pathconc \path_{kj}$,
where $\path_{i\ell} = i \cdots \ell$ and $\path_{kj} = k \cdots j$
do not contain the arcs $(j, i)$ and $(k, \ell)$.
Then
\begin{align*}
  w'(\path_{ij}\cdot i)
    &= w'(\path_{i\ell}) + w'(\ell, k) + w'(\path_{kj}) + w'(j, i) \\
    &= w(\path_{i\ell}) - w_2(k, \ell) - \epsilon
       + w(\path_{kj}) - w_1(i, j) - \epsilon
&\just{def.\ of $w'$} \\
    &\geq w(i, \ell) - w_2(k, \ell) - \epsilon
       + w(k, j) - w_1(i, j) - \epsilon
&\just{$G$ closed} \\
    &= \bigl(
         w(i, \ell) + w(k, j) - w_1(i, j) - w_2(k, \ell)
       \bigr)
         - 2 \epsilon \\
    &\geq 0
&\just{def.\ of $\epsilon$}.
\end{align*}
Therefore $G'$ is consistent.
Moreover, $G' \graphleq G$ since
\begin{align*}
  w'(j, i)
    &= - w_1(i, j) - \epsilon
&\just{def.\ of $w'$} \\
    &\leq - w_1(i, j)
&\just{$\epsilon \geq 0$} \\
    &\leq w_1(j, i)
&\just{$G_1$ consistent} \\
    &\leq w(j, i)
&\just{def.\ $G$};
\end{align*}
similarly, $w'(\ell, k) \leq w(\ell, k)$;
hence, 
for all $r, s \in \cN$, $w'(r,s) \leq w(r,s)$.

Let $\bd' \in \BDshapes_n$ be represented by $G'$,
so that $\emptyset \neq \bd' \sseq \bd_1 \uplus \bd_2$.
Since $w'(j, i) + w_1(i, j) < 0$, we obtain
$\bd' \inters \bd_1 = \emptyset$;
since $w'(\ell, k) + w_2(k, \ell) < 0$,
we obtain $\bd' \inters \bd_2 = \emptyset$.
Hence, $\bd_1 \uplus \bd_2 \neq \bd_1 \union \bd_2$.
\qed
\end{proof}

An algorithm for the exact join detection on $\BDshapes_n$
based on Theorem~\ref{thm:union-of-bds-iff-hull}
will have a worst-case complexity bound in $\bigO(n^4)$.
Noting that the computation of graph closure and reduction
are both in
$\bigO(n^3)$~\cite{BagnaraHMZ05,BagnaraHZ09FMSD,LarsenLPY97,Mine05th},
a more detailed complexity bound is $\bigO(n^3 + r_1r_2)$,
where $r_h$ is the number of arcs in the subgraph $R_h$;
hence, a good choice is to take each $R_h$ to be a graph reduction
for $G_h$, as it will have a minimal number of arcs.

\begin{example}
Consider the 2-dimensional BD shapes
\begin{align*}
  \bd_1 &= \con\bigl(
                 \{ 0 \leq x_1 \leq 3,
                    0 \leq x_2 \leq 2,
                 \}
               \bigr), \\
  \bd_2 &= \con\bigl(
                 \{ 0 \leq x_2 \leq 2, 0 \leq x_1 - x_2 \leq 3 \}
               \bigr)
\end{align*}
shown in Figure~\ref{subfig:shapes:bd1-bd2}.
Then the join $\bd_1 \uplus \bd_2$ is exact.
Note that both conditions~(1) and (2)
in Theorem~\ref{thm:union-of-bds-iff-hull}
play an active role in the decision procedure.
For instance, when taking $i = 1$, $j = 0$, $k = 2$ and $\ell = 1$,
condition~(1) is satisfied but condition~(2) does not hold:
\begin{gather*}
  w_1(1, 0)
    = 3
    < 5
    = w_2(1, 0),
  \qquad
  w_2(2, 1)
    = 0
    < 2
    = w_1(2, 1),\\
  w_1(1, 0) + w_2(2, 1)
    = 3 + 0
    > 0 + 2
    = w(1, 1) + w(2, 0).
\end{gather*}
On the other hand, taking  $i = 1$, $j = 1$, $k = 0$ and $\ell = 2$,
it can be seen that condition~(2) is satisfied but condition~(1)
does not hold:
\begin{gather*}
  w_1(1, 1)
    = 0
    = w_2(1, 1),
  \qquad
  w_2(0, 2)
    = 0
    = w_1(0, 2),\\
  w_1(1, 1) + w_2(0, 2)
    = 0 + 0
    < 3 + 0
    = w(1, 2) + w(0, 1).
\end{gather*}
\end{example}

\subsection{Exact Join Detection for Integer BD Shapes}

We now consider the case of integer BD shapes, i.e.,
subsets of $\Zset^n$ that are delimited by BD constraints
where the bounds are all integral.
As for the rational case,
these numerical abstractions can be encoded using weighted graphs,
but restricting the codomain of the weight function to
$\Zset_\infty \defeq \Zset \union \{ +\infty \}$.
Since the set of all integer graphs is a sub-lattice of
the set of rational graphs,
the conditions in Theorem~\ref{thm:union-of-bds-iff-hull}
can be easily strengthened so as to obtain the corresponding
result for the domain $\IBDshapes_n$ of integer BD shapes.
The complexity bound for the algorithm for the domain of integer BD shapes
is the same as for the rational domain.

\begin{theorem}
\label{thm:union-of-int-bds-iff-hull}
For each $h \in \{ 1, 2 \}$,
let $\bd_h \in \IBDshapes_n$ be a non-empty integer BD shape
represented by the closed integer graph $G_h = (\cN, w_h)$ and
let $R_h$ be a subgraph of $G_h$ such that $\closure(R_h) = G_h$.
Let also $G_1 \graphlub G_2 = (\cN, w)$.
Then $\bd_1 \uplus \bd_2 \neq \bd_1 \union \bd_2$ if and only if
there exist arcs $(i, j)$ of $R_1$ and $(k, \ell)$ of $R_2$
such that
\begin{itemize}
\item[\textup{(1)}]
$w_1(i, j) < w_2(i, j)$ and $w_2(k, \ell) < w_1(k, \ell)$;
and
\item[\textup{(2)}]
$w_1(i, j) + w_2(k, \ell) + 2 \leq w(i, \ell) + w(k, j)$.
\end{itemize}
\end{theorem}
\begin{proof}
Suppose first that $\bd_1 \uplus \bd_2 \neq \bd_1 \union \bd_2$,
so that there exists $\vect{p} \in \Zset^n$ such that
$\vect{p} \in \bd_1 \uplus \bd_2$ but
$\vect{p} \notin \bd_1$ and $\vect{p} \notin \bd_2$.
Hence, there exist $i,j,k,\ell \in \cN$ such that
$(i,j)$ is an arc of $R_1$ satisfying
$\proj_i(\vect{p}) - \proj_j(\vect{p}) > w_1(i, j)$
and $(k, \ell)$ is an arc of $R_2$ satisfying
$\proj_k(\vect{p}) - \proj_\ell(\vect{p}) > w_2(k, \ell)$.
However, as $\vect{p} \in \bd_1 \uplus \bd_2$,
$\proj_i(\vect{p}) - \proj_j(\vect{p}) \leq w(i,j)$ and
$\proj_k(\vect{p}) - \proj_\ell(\vect{p}) \leq w(k, \ell)$
so that, by definition of $G_1 \graphlub G_2$,
we have $w_1(i, j) < w_2(i, j)$ and
$w_2(k, \ell) < w_1(k, \ell)$;
hence condition (1) holds.
Note also that $w_1(i, j)$ and $w_2(k, \ell)$ are both finite
and hence in $\Zset$ so that $w_1(i, j) + 1 \leq w_2(i, j)$ and
$w_2(k, \ell) + 1 \leq w_1(k, \ell)$.
Since $\vect{p} \in \bd_1 \uplus \bd_2$,
\begin{align*}
  w(i, \ell) + w(k, j)
    &\geq \proj_i(\vect{p}) - \proj_\ell(\vect{p})
      + \proj_k(\vect{p}) - \proj_j(\vect{p}) \\
    &= \proj_i(\vect{p}) - \proj_j(\vect{p})
         + \proj_k(\vect{p}) - \proj_\ell(\vect{p}) \\
    &\geq w_1(i, j) + w_2(k, \ell) + 2.
\end{align*}
Therefore, condition (2) also holds.

We now suppose that there exist
arcs $(i,j)$ of $R_1$ and $(k, \ell)$ of $R_2$
such that conditions (1) and (2) hold.
Let $G' = (\cN, w')$ be a graph defined as in the proof of
Theorem~\ref{thm:union-of-bds-iff-hull},
where however we just define $\epsilon \defeq 1$,
so that $G'$ is an integer graph.
By using the same reasoning as in the proof of
Theorem~\ref{thm:union-of-bds-iff-hull}, it can be seen that
$G'$ is consistent and $G' \graphleq G_1 \graphlub G_2$.
Let $\bd' \in \IBDshapes_n$ be represented by $G'$,
so that $\emptyset \neq \bd' \sseq \bd_1 \uplus \bd_2$.
Since $w'(j, i) + w_1(i, j) < 0$,
we obtain
$\bd' \inters \bd_1 = \emptyset$;
since $w'(\ell, k) + w_2(k, \ell) < 0$,
we obtain $\bd' \inters \bd_2 = \emptyset$.
Hence, $\bd_1 \uplus \bd_2 \neq \bd_1 \union \bd_2$.
\qed
\end{proof}

\begin{example}
Consider the 2-dimensional BD shapes
\begin{align*}
  \bd_3 &= \con\bigl(
                 \{ 0 \leq x_1 \leq 3,
                    0 \leq x_2 \leq 2,
                    x_1 - x_2 \leq 2
                 \}
               \bigr), \\
  \bd_4 &= \con\bigl(
                 \{ 3 \leq x_1 \leq 6, 0 \leq x_2 \leq 2 \}
               \bigr)
\end{align*}
shown in Figure~\ref{subfig:shapes:bd3-bd4}.
Then, in the case of rational BD shapes,
the join $\bd_3 \uplus \bd_4$ is not exact;
for instance, letting $\vect{p} = (2.5, 0)^\transpose$
be the point highlighted in Figure~\ref{subfig:shapes:bd3-bd4},
we have $\vect{p} \in \bd_3 \uplus \bd_4$
although $\vect{p} \notin \bd_3 \union \bd_4$.
Taking $i = 1$, $j = 2$, $k = 0$ and $\ell = 1$,
it can be seen that
both conditions in Theorem~\ref{thm:union-of-bds-iff-hull} are satisfied;
in particular, for the second condition we have
\[
  w_1(1, 2) + w_2(0, 1) = 2 - 3
    \leq
      0 + 0 = w(1, 1) + w(0, 2).
\]
By contrast, in the case of integer BD shapes, the join is exact;
all the integral points belonging to the join $\bd_3 \uplus \bd_4$,
denoted by small crosses in Figure~\ref{subfig:shapes:bd3-bd4},
also belong to the union $\bd_3 \union \bd_4$.
In particular, with the above choice for indices $i, j, k, \ell$,
the second condition of Theorem~\ref{thm:union-of-int-bds-iff-hull}
does not hold:
\[
  w_1(1, 2) + w_2(0, 1) + 2 = 2 - 3 + 2
    > 0 + 0 = w(1, 1) + w(0, 2).
\]
\end{example}

\begin{figure}
\centering
\mbox{
\scriptsize{
\subfigure[$\bd_1, \bd_2 \in \BDshapes_2$]{
\label{subfig:shapes:bd1-bd2}
\begin{picture}(120,55)(0,-10)
\put(0,0){
\setlength{\unitlength}{0.7pt}%
\psset{xunit=1cm,yunit=1cm,runit=1cm}
\psset{origin={0,0}}
\psline{->}(-0.5,0)(3.5,0)
\psline{->}(0,-0.3)(0,1.5)
\rput(-0.2,-0.2){O}
\rput(3.2,-0.2){$x_1$}
\rput(-0.2,1.2){$x_2$}
\pspolygon[linecolor=lightgreen,fillcolor=lightgreen,fillstyle=solid](0,0)(1.5,0)(2.5,1)(0,1)
\psline[linecolor=red,linestyle=solid](0,0.01)(1.5,0.01)(1.5,1.01)(0,1.01)(0,0)
\psline[linecolor=blue,linestyle=solid](0,-0.01)(1.5,-0.01)(2.5,0.99)(1,0.99)(0,-0.01)
\rput(0.35,0.75){$\bd_1$}
\rput(1.85,0.75){$\bd_2$}
}
\end{picture}
}
\qquad
\subfigure[$\bd_3, \bd_4 \in \BDshapes_2$]{
\label{subfig:shapes:bd3-bd4}
\begin{picture}(120,55)(0,-10)
\put(0,0){
\setlength{\unitlength}{0.7pt}%
\psset{xunit=1cm,yunit=1cm,runit=1cm}
\psset{origin={0,0}}
\psline{->}(-0.5,0)(3.5,0)
\psline{->}(0,-0.3)(0,1.5)
\rput(-0.2,-0.2){O}
\rput(3.2,-0.2){$x_1$}
\rput(-0.2,1.2){$x_2$}
\pspolygon[linecolor=lightgreen,fillcolor=lightgreen,fillstyle=solid](0,0)(1,0)(1.5,0.5)(1.5,1)(0,1)
\pspolygon[linecolor=lightgreen,fillcolor=lightgreen,fillstyle=solid](1.5,0)(3,0)(3,1)(1.5,1)
\psline[linecolor=red,linestyle=solid](0,0)(1,0)
\psline[linecolor=red,linestyle=solid](1,0)(1.5,0.5)
\psline[linecolor=red,linestyle=solid,linewidth=0.5pt](1.49,0.5)(1.49,1)
\psline[linecolor=red,linestyle=solid](1.49,1)(0,1)
\psline[linecolor=red,linestyle=solid](0,1)(0,0)
\psline[linecolor=blue,linestyle=solid](1.51,0)(3,0)
\psline[linecolor=blue,linestyle=solid](3,0)(3,1)
\psline[linecolor=blue,linestyle=solid](3,1)(1.5,1)
\psline[linecolor=blue,linestyle=solid,linewidth=0.5pt](1.51,1)(1.51,0)
\psdot[linecolor=dgreen,dotstyle=BoldMul](0,0)
\psdot[linecolor=dgreen,dotstyle=BoldMul](0,0)
\psdot[linecolor=dgreen,dotstyle=BoldMul](0,0.5)
\psdot[linecolor=dgreen,dotstyle=BoldMul](0,1)
\psdot[linecolor=dgreen,dotstyle=BoldMul](0.5,0)
\psdot[linecolor=dgreen,dotstyle=BoldMul](0.5,0.5)
\psdot[linecolor=dgreen,dotstyle=BoldMul](0.5,1)
\psdot[linecolor=dgreen,dotstyle=BoldMul](1,0)
\psdot[linecolor=dgreen,dotstyle=BoldMul](1,0.5)
\psdot[linecolor=dgreen,dotstyle=BoldMul](1,1)
\psdot[linecolor=dgreen,dotstyle=BoldMul](1.5,0)
\psdot[linecolor=dgreen,dotstyle=BoldMul](1.5,0.5)
\psdot[linecolor=dgreen,dotstyle=BoldMul](1.5,1)
\psdot[linecolor=dgreen,dotstyle=BoldMul](2,0)
\psdot[linecolor=dgreen,dotstyle=BoldMul](2,0.5)
\psdot[linecolor=dgreen,dotstyle=BoldMul](2,1)
\psdot[linecolor=dgreen,dotstyle=BoldMul](2.5,0)
\psdot[linecolor=dgreen,dotstyle=BoldMul](2.5,0.5)
\psdot[linecolor=dgreen,dotstyle=BoldMul](2.5,1)
\psdot[linecolor=dgreen,dotstyle=BoldMul](3,0)
\psdot[linecolor=dgreen,dotstyle=BoldMul](3,0.5)
\psdot[linecolor=dgreen,dotstyle=BoldMul](3,1)
\pscircle*[linecolor=dgreen](1.25,0){1.5pt}
\rput(1.25,-0.2){$\vect{p}$}
\rput(0.75,0.75){$\bd_3$}
\rput(2.25,0.75){$\bd_4$}
}
\end{picture}
}
}
}
\caption{Examples for the join of rational and integer BD shapes}
\label{fig:shapes-examples}
\end{figure}
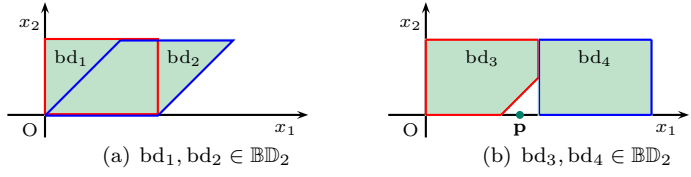

\subsection{Generalizing to $k$ BD shapes}

We conjecture that the above results for the exact join
detection of two (rational or integer) BD shapes can be generalized
to any number of component BD shapes.
That is, given $k$ BD shapes $\bd_1, \ldots, \bd_k \in \BDshapes_n$,
it is possible to provide a suitable set of conditions that determine
whether or not
$\bd_1 \uplus \cdots \uplus \bd_k = \bd_1 \union \cdots \union \bd_k$.
Here we just present the conjecture, for the rational case, when $k = 3$.

\begin{conjecture}
\label{conj:union-of-3-bds-iff-hull}
For each $h \in \{ 1, 2, 3 \}$,
let $\bd_h \in \BDshapes_n$ be a non-empty BD shape
represented by the closed graph $G_h = (\cN, w_h)$ and
let $R_h$ be a subgraph of $G_h$ such that $\closure(R_h) = G_h$.
Let also $G_1 \graphlub G_2 \graphlub G_3 = (\cN, w)$.
Then $\bd_1 \uplus \bd_2 \uplus \bd_3 \neq \bd_1 \union \bd_2 \union \bd_3$
if and only if there exist arcs $(i_1, j_1)$ of $R_1$,
 $(i_2, j_2)$ of $R_2$ and $(i_3, j_3)$ of $R_3$, respectively,
such that
\begin{itemize}
\item[\textup{(1)}]
for each $h \in \{ 1, 2, 3 \}$,
$w_h(i_h, j_h) < w(i_h, j_h)$;
\item[\textup{(2a)}]
$w_1(i_1, j_1) + w_2(i_2, j_2) < w(i_1, j_2) + w(i_2, j_1)$;
\item[\textup{(2b)}]
$w_2(i_2, j_2) + w_3(i_3, j_3) < w(i_2, j_3) + w(i_3, j_2)$;
\item[\textup{(2c)}]
$w_3(i_3, j_3) + w_1(i_1, j_1) < w(i_3, j_1) + w(i_1, j_3)$;
\item[\textup{(3a)}]
\(
  w_1(i_1, j_1) + w_2(i_2, j_2) + w_3(i_3, j_3)
    < w(i_1, j_2) + w(i_2, j_3) + w(i_3, j_1)
\);
\item[\textup{(3b)}]
\(
  w_1(i_1, j_1) + w_2(i_2, j_2) + w_3(i_3, j_3)
    < w(i_1, j_3) + w(i_2, j_1) + w(i_3, j_2)
\).
\end{itemize}
\end{conjecture}

Even though the generalization is straightforward from a mathematical
point of view, for larger values of $k$ this will result in having to
check a rather involved combinatorial combination of all the conditions.

\section{Exact Join Detection for Octagonal Shapes}
\label{sec:oct-shapes}

Octagonal constraints generalize BD constraints by also
allowing for non-strict inequalities having the form
$x_i + x_j \leq b$ or $-x_i - x_j \leq b$.
This class of constraints was first proposed in~\cite{BalasundaramK89}
and further elaborated in \cite{Mine01b}.

\subsection{Octagonal Shapes and their Graph Representation}

We first introduce the required notation and terminology
(see also~\cite{BagnaraHMZ05,BagnaraHZ09FMSD,Mine05th}).

Octagonal constraints can be encoded using BD constraints
by splitting each variable $x_i$ into two forms:
a positive form $x_i^+$, interpreted as $+x_i$;
and a negative form $x_i^-$, interpreted as $-x_i$.
For instance, an octagonal constraint such as $x_i + x_j \leq b$ can be
translated into the potential constraint
$x_i^+ - x_j^- \leq b$; alternatively, the same octagonal constraint
can be translated into $x_j^+ - x_i^- \leq b$.
Unary (octagonal) constraints such as
$x_i \leq b$ and $-x_i \leq b$
are encoded as $x_i^+ - x_i^- \leq 2b$
and $x_i^- - x_i^+ \leq -2b$, respectively.

From now on, we assume that the set of nodes
is $\Vpm \defeq \{ 0, \ldots, 2n-1 \}$.
These will denote the positive and negative forms
of the vector space dimensions $x_1$, \ldots, $x_n$:
for all $i \in \Vpm$, if $i = 2k$, then
$i$ represents the positive form $x_{k+1}^+$ and, if $i = 2k+1$, then
$i$ represents the negative form $x_{k+1}^-$ of the dimension $x_{k+1}$.
To simplify the presentation,
we let $\bari$ denote $i+1$, if $i$ is even,
and $i-1$, if $i$ is odd, so that, for all $i \in \Vpm$, we also have
$\bari \in \Vpm$ and $\bar{\bari} = i$.

It follows from the above translations that
any finite system of octagonal constraints,
translated into a set of potential constraints in $\Vpm$ as above,
can be encoded by a graph $G$ in $\Vpm$.
In particular, any finite \emph{satisfiable} system
of octagonal constraints can be encoded by a \emph{consistent}
graph in $\Vpm$.
However, the converse does not hold
since in any valuation $\rho$ of an encoding of
a set of octagonal constraints
we must also have $\rho(i) = -\rho(\bari)$, so that
the arcs $(i,j)$ and $(\barj, \bari)$ should have the same weight.
Therefore, to encode rational octagonal constraints,
we restrict attention to consistent graphs over $\Vpm$
where the arcs in all such pairs are \emph{coherent}.

\begin{definition}
\label{def:octagonal-graph}
\summary{(Octagonal graph.)}
A (rational) \emph{octagonal graph} is
any consistent graph $G =(\Vpm, w)$ that
satisfies the coherence assumption:
\begin{equation}
\label{rule:coherence}
  \forall i,j \in \Vpm
    \itc w(i,j) = w(\barj, \bari).
\end{equation}
\end{definition}
The set $\Octgraphs$ of all octagonal graphs
(with the usual addition of the bottom element,
representing an unsatisfiable system of constraints)
is a sub-lattice of $\Graphs_\bot$, sharing the same
least upper bound and greatest lower bound operators.
Note that, at the implementation level, coherence can be
automatically and efficiently enforced by letting arc $(i,j)$
and arc $(\barj, \bari)$ share the same representation.


The standard shortest-path closure algorithm is not enough to obtain
a canonical form for octagonal graphs.

\begin{definition}
\label{def:strong-closure}
\summary{(Graph strong closure/reduction.)}
An octagonal graph $G = (\Vpm, w)$ is \emph{strongly closed}
if it is closed and the following property holds:
\begin{equation}
\label{rule:strong-coherence}
  \forall i,j \in \Vpm
    \itc 2 w(i,j) \leq w(i, \bari) + w(\barj, j).
\end{equation}
The \emph{strong closure} of an octagonal graph $G$ in $\Vpm$ is
\[
  \strongclosure(G)
    \defeq
      \biggraphlub
        \bigl\{\,
          G' \in \Octgraphs
        \bigm|
          \text{$G' \graphleq G$ and $G'$ is strongly closed\/}
        \,\bigr\}.
\]
An octagonal graph $R$ is \emph{strongly reduced} if,
for each octagonal graph $G \neq R$ such that $R \graphleq G$,
we have $\strongclosure(R) \neq \strongclosure(G)$.
A \emph{strong reduction} for the octagonal graph $G$ is
any strongly reduced octagonal graph $R$
such that $\strongclosure(R) = \strongclosure(G)$.
\end{definition}
Observe that, as was the case for shortest-path reduction,
a strong reduction for a strongly closed graph $G$
is a subgraph of $G$.

We denote by $\Octshapes_n$ the domain of octagonal shapes,
whose non-empty elements can be represented by octagonal graphs:
$\BDshapes_n$ is a meet-sublattice of $\Octshapes_n$,
which in turn is a meet-sublattice of $\CPset_n$.
A strongly closed (resp., strongly reduced) graph encodes
a system of octagonal constraints which is closed by entailment
(resp., contains no redundant constraint).

\subsection{Exact Join Detection for Rational Octagonal Shapes}

An exact join decision procedure specialized for rational octagonal shapes
can be based on the following result.

\begin{theorem}
\label{thm:union-of-octs-iff-hull}
For each $h \in \{ 1, 2 \}$,
let $\oct_h \in \Octshapes_n$ be a non-empty octagonal shape
represented by the strongly closed graph $G_h = (\cN, w_h)$ and
let $R_h$ be a subgraph of $G_h$ such that $\strongclosure(R_h) = G_h$.
Let also $G_1 \graphlub G_2 = (\Vpm, w)$.
Then $\oct_1 \uplus \oct_2 \neq \oct_1 \union \oct_2$ if and only if
there exist arcs $(i, j)$ of $R_1$ and $(k, \ell)$ of $R_2$
such that
\begin{itemize}
\item[\textup{(1a)}]
$w_1(i, j) < w_2(i, j)$;
\item[\textup{(1b)}]
$w_2(k, \ell) < w_1(k, \ell)$;
\item[\textup{(2a)}]
$w_1(i, j) + w_2(k, \ell) < w(i, \ell) + w(k, j)$;
\item[\textup{(2b)}]
$w_1(i, j) + w_2(k, \ell) < w(i, \bark) + w(\barj, \ell)$;
\item[\textup{(3a)}]
$2w_1(i, j) + w_2(k, \ell) < w(i, \ell) + w(i, \bark) + w(\barj, j)$;
\item[\textup{(3b)}]
$2w_1(i, j) + w_2(k, \ell) < w(k, j) + w(\barj, \ell) + w(i, \bari)$;
\item[\textup{(4a)}]
$w_1(i, j) + 2w_2(k, \ell) < w(i, \ell) + w(\barj, \ell) + w(k, \bark)$;
\item[\textup{(4b)}]
$w_1(i, j) + 2w_2(k, \ell) < w(k, j) + w(i, \bark) + w(\barl, \ell)$.
\end{itemize}
\end{theorem}
\begin{proof}
For each $r \in \Vpm = \{ 0, \ldots, 2n-1 \}$ and
each $\vect{v} = (v_1, \ldots, v_n)^\transpose \in \Rset^n$,
we denote by $\octproj_r(\vect{v})$ the projection of vector $\vect{v}$
on the space dimension corresponding to the octagonal graph node $r$,
defined as:
\[
  \octproj_r(\vect{v})
    \defeq
      \begin{cases}
          v_{s+1},  &\text{if $r = 2s$;} \\
        - v_{s+1},  &\text{if $r = 2s+1$.}
      \end{cases}
\]

Suppose that $\oct_1 \uplus \oct_2 \neq \oct_1 \union \oct_2$,
so that there exists $\vect{p} \in \oct_1 \uplus \oct_2$
such that $\vect{p} \notin \oct_1$ and $\vect{p} \notin \oct_2$.
Hence, there exist arcs
$(i,j)$ and $(k, \ell)$ of $R_1$ and $R_2$, respectively, satisfying
\begin{align*}
  w(i, j)
    &\geq \octproj_i(\vect{p}) - \octproj_j(\vect{p})
    > w_1(i, j), \\
  w(k, \ell)
    &\geq \octproj_k(\vect{p}) - \octproj_\ell(\vect{p})
    > w_2(k, \ell); \\
\intertext{%
hence conditions (1a) and (1b) hold;}
  w(i, \ell) + w(k, j)
    &\geq \octproj_i(\vect{p}) - \octproj_\ell(\vect{p})
      + \octproj_k(\vect{p}) - \octproj_j(\vect{p}) \\
    &= \octproj_i(\vect{p}) - \octproj_j(\vect{p})
         + \octproj_k(\vect{p}) - \octproj_\ell(\vect{p}) \\
    &> w_1(i, j) + w_2(k, \ell) \\
\intertext{%
so that condition (2a) holds and, by a symmetric argument,
condition (2b) holds;}
  w(i, \ell) + w(i, \bark) + w(\barj, j)
    &\geq \bigl( \octproj_i(\vect{p}) - \octproj_\ell(\vect{p}) \bigr)
      + \bigl( \octproj_i(\vect{p}) + \octproj_k(\vect{p}) \bigr)
      + \bigl( - 2 \octproj_j(\vect{p}) \bigr) \\
    &= 2 \bigl( \octproj_i(\vect{p}) - \octproj_j(\vect{p}) \bigr)
         + \octproj_k(\vect{p}) - \octproj_\ell(\vect{p}) \\
    &> 2 w_1(i, j) + w_2(k, \ell)
\end{align*}
so that condition (3a) holds;
conditions (3b), (4a) and (4b) follow by symmetric arguments.

We now suppose that, for some $i, j, k, \ell \in \cN$,
all conditions (1a) -- (4b) hold.
Note that,
by (1a) and (1b),
$i \neq j$ and $k \neq \ell$.
Suppose first that
$(i, j) \in \bigl\{ (\ell, k), (\bark, \barl) \bigr\}$;
then, conditions (2a) and (2b)
imply $w_1(i, j) + w_2(j, i) < 0$,
so that the graph greatest lower bound $G_1 \graphglb G_2$ is inconsistent,
as it contains a negative weight proper cycle;
hence, $\oct_1 \inters \oct_2 = \emptyset$,
which implies $\oct_1 \uplus \oct_2 \neq \oct_1 \union \oct_2$.
Therefore, in the following we assume that
$(i, j) \notin \bigl\{ (\ell, k), (\bark, \barl) \bigr\}$ holds.

If the right hand sides of the inequalities in
conditions (1a) -- (4b) are all unbounded,
let $\epsilon \defeq 1$;
otherwise let
\[
  \epsilon
    \defeq
      \min\left\{
        \begin{aligned}
          &w(i, j) - w_1(i, j), \\
          &w(k, \ell) - w_2(k, \ell), \\
          &\frac{1}{2}
             \bigl(
               w(i, \ell) + w(k, j) - w_1(i, j) - w_2(k, \ell)
             \bigr),\\
          &\frac{1}{2}
             \bigl(
               w(i, \bark) + w(\barj, \ell) - w_1(i, j) - w_2(k, \ell)
             \bigr),\\
          &\frac{1}{3}
             \bigl(
               w(i, \ell) + w(i, \bark) + w(\barj, j)
                 - 2w_1(i, j) - w_2(k, \ell)
             \bigr),\\
          &\frac{1}{3}
             \bigl(
               w(k, j) + w(\barj, \ell) + w(i, \bari)
                 - 2w_1(i, j) - w_2(k, \ell)
             \bigr),\\
          &\frac{1}{3}
             \bigl(
               w(i, \ell) + w(\barj, \ell) + w(k, \bark)
                 - w_1(i, j) - 2w_2(k, \ell)
             \bigr),\\
          &\frac{1}{3}
             \bigl(
               w(k, j) + w(i, \bark) + w(\barl, \ell)
                 - w_1(i, j) - 2w_2(k, \ell)
             \bigr)
        \end{aligned}
      \right\}.
\]
Then, by conditions (1a) -- (4b)
$\epsilon > 0$.
Consider the graph $G' = (\cN, w')$ where,
for each $r, s \in \cN$,
\begin{equation*}
  w'(r, s)
    \defeq
      \begin{cases}
        - w_1(i, j) - \epsilon,
          &\text{if $(r, s) \in \bigl\{(j, i), (\bari, \barj)\bigr\}$;} \\
        - w_2(k, \ell) - \epsilon,
          &\text{if $(r, s) \in \bigl\{(\ell, k), (\bark, \barl)\bigr\}$;} \\
        w(r, s),
          &\text{otherwise.}
      \end{cases}
\end{equation*}
Let $G \defeq G_1 \graphlub G_2$;
as $G$ is coherent, $G'$ is coherent too.
We now show that $G'$ is a consistent graph;
to this end, since $G$ is known to be consistent,
it is sufficient to consider the proper cycles of $G'$
that contain arc $(j, i)$ or arc $(\ell, k)$.
\footnote{Any cycle containing arc $(\bari, \barj)$
(resp., $(\bark, \barl)$) can be transformed to the corresponding
coherent cycle containing arc $(j, i)$ (resp., $(\ell, k)$),
having the same weight.}
Let $\path_{ij} = i\cdots j$ and $\path_{k\ell} = k\cdots \ell$ be
any simple paths from $i$ to $j$ and from $k$ to $\ell$, respectively.
Then $G'$ is consistent if and only if
$w'(\path_{ij} \cdot i) \geq 0$ and
$w'(\path_{k\ell} \cdot k) \geq 0$.
We only prove $w'(\path_{ij} \cdot i) \geq 0$
since the proof that
$w'(\path_{k\ell} \cdot k) \geq 0$
follows by a symmetrical argument.
Since $\theta_{ij}$ is simple, it does not contain the arc $(j, i)$.
In the following we consider in detail five cases,
again noting that all the other cases
can be proved by symmetrical arguments:
\begin{enumerate}
\item
$\path_{ij}$ contains none of the arcs
$(\ell, k)$, $(\bark, \barl)$ and $(\bari, \barj)$;
\item
\(
  \path_{ij}
    = \path_{i\bari} \pathconc (\bari, \barj) \pathconc \path_{\barj j}
\);
\item
\(
  \path_{ij}
    = \path_{i\ell} \pathconc (\ell, k) \pathconc \path_{kj}
\);
\item
\(
  \path_{ij}
    = \path_{i\ell} \pathconc (\ell, k) \pathconc \path_{k\bark}
        \pathconc (\bark, \barl) \pathconc \path_{\barl j}
\);
\item
\(
  \path_{ij}
    = \path_{i\ell} \pathconc (\ell, k) \pathconc \path_{k\bark}
        \pathconc (\bark, \barl) \pathconc \path_{\barl \bari}
          \pathconc (\bari, \barj) \pathconc \path_{\barj j},
\)
\end{enumerate}
where the simple paths
$\path_{i\bari}$, $\path_{i\ell}$, $\path_{kj}$, $\path_{k\bark}$,
$\path_{\barl j}$, $\path_{\barl \bari}$ and $\path_{\barj j}$
contain none of the arcs $(\ell, k)$, $(\bark, \barl)$ and $(\bari, \barj)$.
\begin{itemize}
\item Case (1).
\begin{align*}
  w'(\path_{ij}\cdot i)
    &= w'(\path_{ij}) + w'(j, i) \\
    &= w(\path_{ij}) - w_1(i, j) - \epsilon
&\just{def.\ of $w'$} \\
    &\geq w(i, j) - w_1(i, j) - \epsilon
&\just{$G$ closed} \\
    &\geq 0
&\just{def.\ of $\epsilon$}.
\end{align*}
\item Case (2).
\begin{align*}
  w'(\path_{ij}\cdot i)
    &= w'(\path_{i\bari}) + w'(\bari, \barj) + w'(\path_{\barj j}) + w'(j, i) \\
    &= w'(\path_{i\bari}) + w'(\path_{\barj j}) + 2 w'(j, i)
&\just{$G'$ coherent} \\
    &= w(\path_{i\bari}) + w(\path_{\barj j}) - 2 w_1(i, j) - 2 \epsilon
&\just{def.\ of $w'$} \\
    &\geq w(i, \bari) + w(\barj, j) - 2 w_1(i, j) - 2 \epsilon
&\just{$G$ closed} \\
    &\geq 2 w(i, j) - 2 w_1(i, j) - 2 \epsilon
&\just{$G$ strongly closed} \\
    &= 2 \bigl( w(i, j) - w_1(i, j) \bigr) - 2 \epsilon \\
    &\geq 0
&\just{def.\ of $\epsilon$}.
\end{align*}
\item Case (3).
\begin{align*}
  w'(\path_{ij}\cdot i)
    &= w'(\path_{i\ell}) + w'(\ell, k) + w'(\path_{kj}) + w'(j, i) \\
    &= w(\path_{i\ell}) - w_2(k, \ell) - \epsilon
         + w(\path_{kj}) - w_1(i, j) - \epsilon
&\just{def.\ of $w'$} \\
    &\geq w(i, \ell) - w_2(k, \ell) - \epsilon
            + w(k, j) - w_1(i, j) - \epsilon
&\just{$G$ closed} \\
    &= \bigl(
         w(i, \ell) + w(k, j) - w_1(i, j) - w_2(k, \ell)
       \bigr)
         - 2 \epsilon \\
    &\geq 0
&\just{def.\ of $\epsilon$}.
\end{align*}
\item Case (4).
\begin{align*}
& w'(\path_{ij}\cdot i) \\
    &= w'(\path_{i\ell}) + w'(\ell, k) + w'(\path_{k\bark})
         + w'(\bark, \barl) + w'(\path_{\barl j}) + w'(j, i) \\
    &= w'(\path_{i\ell}) + 2 w'(\ell, k) + w'(\path_{k\bark})
         + w'(\path_{\barj \ell}) + w'(j, i)
&\just{$G'$ coherent} \\
    &= w(\path_{i\ell}) - 2 w_2(k, \ell) - 2 \epsilon + w(\path_{k\bark})
         + w(\path_{\barj \ell}) - w_1(i, j) - \epsilon
&\just{def.\ of $w'$} \\
    &\geq w(i, \ell) - 2 w_2(k, \ell) - 2 \epsilon + w(k, \bark)
         + w(\barj, \ell) - w_1(i, j) - \epsilon
&\just{$G$ closed} \\
    &= \bigl(
         w(i, \ell) + w(\barj, \ell) + w(k, \bark)
          - w_1(i, j) - 2 w_2(k, \ell)
       \bigr)
         - 3 \epsilon \\
    &\geq 0
&\just{def.\ of $\epsilon$}.
\end{align*}
\item Case (5).
\begin{align*}
  w'(\path_{ij}\cdot i)
    &= w'(\path_{i\ell}) + w'(\ell, k)
       + w'(\path_{k\bark}) + w'(\bark, \barl) \\
    &\qquad
       + w'(\path_{\barl \bari}) + w'(\bari, \barj)
       + w'(\path_{\barj j}) + w'(j, i) \\
    &= 2 w'(\path_{i\ell}) + 2 w'(j, i) \\
    &\qquad
       + w'(\path_{k\bark}) + w'(\path_{\barj j}) + 2 w'(\ell, k)
&\just{$G'$ coherent} \\
    &= 2 w(\path_{i\ell}) - 2 w_1(i, j) - 2 \epsilon \\
    &\qquad
       + w(\path_{k\bark}) + w(\path_{\barj j}) - 2 w_2(k, \ell) - 2 \epsilon
&\just{def.\ of $w'$} \\
    &\geq 2 w(i, \ell) - 2 w_1(i, j) - 2 \epsilon \\
    &\qquad
       + w(k, \bark) + w(\barj, j) - 2 w_2(k, \ell) - 2 \epsilon
&\just{$G$ closed} \\
    &\geq 2 w(i, \ell) - 2 w_1(i, j) - 2 \epsilon \\
    &\qquad
       + 2 w(k, j) - 2 w_2(k, \ell) - 2 \epsilon
&\just{$G$ strongly closed} \\
    &= 2 \bigl(w(i, \ell) + w(k, j) 
       - w_1(i, j) - w_2(k, \ell)\bigr) - 4 \epsilon \\
    &\geq 0
&\just{def.\ of $\epsilon$}.
\end{align*}
\end{itemize}
Therefore $G'$ is consistent.
Moreover, $G' \graphleq G$ since
\begin{align*}
  w'(j, i)
    &= - w_1(i, j) - \epsilon
&\just{def.\ of $w'$} \\
    &\leq - w_1(i, j)
&\just{$\epsilon \geq 0$} \\
    &\leq w_1(j, i)
&\just{$G_1$ consistent} \\
    &\leq w(j, i)
&\just{def.\ $G$};
\end{align*}
similarly, $w'(\ell, k) \leq w(\ell, k)$;
hence, for all $r, s \in \cN$, $w'(r,s) \leq w(r,s)$.

Let $\oct' \in \Octshapes_n$ be represented by $G'$,
so that $\emptyset \neq \oct' \sseq \oct_1 \uplus \oct_2$.
Since $w'(j, i) + w_1(i, j) < 0$, we obtain
$\oct' \inters \oct_1 = \emptyset$;
since $w'(\ell, k) + w_2(k, \ell) < 0$, we obtain
$\oct' \inters \oct_2 = \emptyset$.
Hence, $\oct_1 \uplus \oct_2 \neq \oct_1 \union \oct_2$.
\qed
\end{proof}

Since the computation of the strong closure and strong reduction
of an octagonal graph are both in $\bigO(n^3)$
\cite{BagnaraHMZ05,BagnaraHZ09FMSD,Mine05th},
an algorithm for the exact join detection on $\Octshapes_n$
based on Theorem~\ref{thm:union-of-octs-iff-hull} has the same asymptotic
worst-case complexity as the corresponding algorithm for $\BDshapes_n$.

\begin{example}
Consider the 2-dimensional octagonal shapes
\begin{align*}
  \oct_1 &= \con\bigl(
                  \{ x_1 + x_2 \leq 0 \}
                \bigr), \\
  \oct_2 &= \con\bigl(
                  \{ x_1 \leq 2 \}
                \bigr\}.
\end{align*}
Then the join $\oct_1 \uplus \oct_2 = \Rset^2$ is not exact.
Taking the nodes $i = 0$, $j = 3$, $k = 0$ and $\ell = 1$
(which represent the signed form variables
$x_1^+$, $x_2^-$, $x_1^+$ and $x_1^-$, respectively),
we have $w_1(i, j) = 0$ (encoding $x_1 + x_2 \leq 0$)
and $w_2(k, \ell) = 4$ (encoding $x_1 + x_1 \leq 4$, i.e., $x_1 \leq 2$).
So all the left hand sides in conditions (1a) -- (4b) are finite
while all the corresponding right hand sides are infinite;
and hence all the conditions will hold.
\end{example}

\subsection{Exact Join Detection for Integer Octagonal Shapes}

We now consider the case of integer octagonal constraints, i.e.,
octagonal constraints where the bounds are all integral and the variables
are only allowed to take integral values.
These can be encoded by suitably restricting the codomain of the
weight function of octagonal graphs.

\begin{definition}
\label{def:integer-octagonal-graph}
\summary{(Integer octagonal graph.)}
An \emph{integer octagonal graph} is an octagonal graph $G = (\Vpm, w)$
having an integral weight function:
\[
  \forall i,j \in \Vpm \itc w(i,j) \in \Zset \union \{ +\infty \}.
\]
\end{definition}

As an integer octagonal graph is also a rational octagonal graph,
the constraint system that it encodes
will be satisfiable when interpreted to take values in $\Qset$.
However, when interpreted to take values in $\Zset$,
this system may be unsatisfiable
since the arcs encoding unary constraints can have an odd weight;
we say that an octagonal graph is \emph{$\Zset$-consistent}
if its encoded integer constraint system is satisfiable.
For the same reason, the strong closure of an integer octagonal graph does
not provide a canonical form for the integer constraint system.

\begin{definition}
\label{def:tight-closure}
\summary{(Graph tight closure/reduction.)}
An octagonal graph $G = (\Vpm, w)$ is \emph{tightly closed}
if it is a strongly closed integer octagonal graph
and the following property holds:
\begin{equation}
\label{rule:tight-coherence}
  \forall i \in \Vpm
    \itc \mathord{\text{$w(i, \bari)$ is even}}.
\end{equation}
The \emph{tight closure} of an octagonal graph $G$ in $\Vpm$ is
\[
  \tightclosure(G)
    \defeq
      \biggraphlub
        \bigl\{\,
          G' \in \Octgraphs
        \bigm|
          \text{$G' \graphleq G$ and $G'$ is tightly closed\/}
        \,\bigr\}.
\]
A $\Zset$-consistent integer octagonal graph $R$ is \emph{tightly reduced} if,
for each integer octagonal graph $G \neq R$ such that $R \graphleq G$,
we have $\tightclosure(R) \neq \tightclosure(G)$.
A \emph{tight reduction} for the
$\Zset$-consistent integer octagonal graph $G$ is
any tightly reduced graph $R$ such that
$\tightclosure(R) = \tightclosure(G)$.
\end{definition}
It follows from these definitions that
any tightly closed integer octagonal graph
encodes a satisfiable integer constraint system
if and only if it is $\Zset$-consistent~\cite{BagnaraHZ08,BagnaraHZ09FMSD}.
Therefore, tight closure is a kernel operator on the lattice of
octagonal graphs, as was the case for strong closure.
Observe also that a tight reduction for a tightly closed graph $G$
is a subgraph of $G$~\cite{BagnaraHZ09FMSD}.
We denote by $\IOctshapes_n$ the domain of integer octagonal shapes.

To prove the Theorem~\ref{thm:union-of-int-octs-iff-hull} below,
we will also use the following result
proved in~\cite[Lemma 4]{LahiriM05}.
\begin{lemma}
\label{lem:LahiriM05-lemma-4}
Let $G =(\Vpm, w)$ be an integer octagonal graph
with no negative weight cycles and
$G_\rt = (\Vpm, w_\rt)$
be a graph having a negative weight cycle and
such that $w_\rt$ satisfies
\[
  w_\rt(i, j)
    \defeq
      \begin{cases}
        2 \lfloor w(i, j)/2 \rfloor,
          &\text{if $j = \bari$;} \\
        w(i, j),
          &\text{otherwise.}
      \end{cases}
\]
Then there exist $i,\bari \in \Vpm$ and a cycle
$\pi = (i \cdot \pi_1 \cdot \bari) \pathconc (\bari  \cdot\pi_2 \cdot i)$
in $G$ such that $w(\pi) = 0$ and the weight of the shortest path in
$G$ from $i$ to $\bari$ is odd.
\end{lemma}

We are now ready to state the condition for exact join detection
for integer octagonal shapes.

\begin{theorem}
\label{thm:union-of-int-octs-iff-hull}
For each $h \in \{ 1, 2 \}$,
let $\oct_h \in \IOctshapes_n$ be a non-empty integer octagonal shape
represented by the tightly closed graph $G_h = (\cN, w_h)$ and
let $R_h$ be a subgraph of $G_h$ such that $\tightclosure(R_h) = G_h$.
Let also $G_1 \graphlub G_2 = (\Vpm, w)$.
Then $\oct_1 \uplus \oct_2 \neq \oct_1 \union \oct_2$ if and only if
there exists arcs $(i, j)$ of $R_1$ and $(k, \ell)$ of $R_2$
such that,
letting $\epsilon_{ij} = 2$ if $j = \bari$ and $\epsilon_{ij} = 1$
otherwise
and $\epsilon_{k\ell} = 2$ if $\ell = \bark$ and $\epsilon_{k\ell} = 1$
otherwise, the following hold:
\begin{itemize}
\item[\textup{(1a)}]
$w_1(i, j) + \epsilon_{ij} \leq w_2(i, j)$;
\item[\textup{(1b)}]
$w_2(k, \ell) + \epsilon_{k\ell} \leq w_1(k, \ell)$;
\item[\textup{(2a)}]
\(
  w_1(i, j) + w_2(k, \ell) + \epsilon_{ij} + \epsilon_{k\ell}
    \leq w(i, \ell) + w(k, j)
\);
\item[\textup{(2b)}]
\(
  w_1(i, j) + w_2(k, \ell) + \epsilon_{ij} + \epsilon_{k\ell}
    \leq w(i, \bark) + w(\barl, j)
\);
\item[\textup{(3a)}]
\(
  2w_1(i, j) + w_2(k, \ell) + 2\epsilon_{ij} + \epsilon_{k\ell}
    \leq w(i, \ell) + w(k, \bari) + w(\barj, j)
\);
\item[\textup{(3b)}]
\(
  2w_1(i, j) + w_2(k, \ell) + 2\epsilon_{ij} + \epsilon_{k\ell}
    \leq w(k, j) + w(\barj, \ell) + w(i, \bari)
\);
\item[\textup{(4a)}]
\(
  w_1(i, j) + 2w_2(k, \ell) + \epsilon_{ij} + 2\epsilon_{k\ell}
    \leq w(k, j) + w(i, \bark) + w(\barl, \ell)
\);
\item[\textup{(4b)}]
\(
  w_1(i, j) + 2w_2(k, \ell) + \epsilon_{ij} + 2\epsilon_{k\ell}
    \leq w(i, \ell) + w(\barl, j) + w(k, \bark)
\).
\end{itemize}
\end{theorem}
\begin{proof}
We will use the notation $\octproj$ as defined in the proof of
Theorem~\ref{thm:union-of-octs-iff-hull}.
Suppose that $\oct_1 \uplus \oct_2 \neq \oct_1 \union \oct_2$,
so that there exists $\vect{p} \in \oct_1 \uplus \oct_2$
such that $\vect{p} \notin \oct_1$ and $\vect{p} \notin \oct_2$.
Hence, letting
$\tilde{p}_{ij} \defeq \octproj_i(\vect{p}) - \octproj_j(\vect{p})$ and
$\tilde{p}_{k\ell} \defeq \octproj_k(\vect{p}) - \octproj_\ell(\vect{p})$,
there exist arcs
$(i,j)$ and $(k, \ell)$ of $R_1$ and $R_2$, respectively, satisfying
\(
  \tilde{p}_{ij} > w_1(i, j)
\)
and
\(
  \tilde{p}_{k\ell} > w_2(k, \ell)
\);
as $\vect{p} \in \oct_1 \uplus \oct_2$, we also have
\(
  w_2(i, j) \geq \tilde{p}_{ij}
\)
and
\(
  w_1(k, \ell) \geq \tilde{p}_{k\ell}
\).
Note that $w_1(i, j)$ and $w_2(k, \ell)$ are both finite
and hence in $\Zset$ so that $\tilde{p}_{ij} \geq w_1(i, j) + 1$ and
$\tilde{p}_{k\ell} \geq w_2(k, \ell) + 1$;
also, by the tight coherence rule~\eqref{rule:tight-coherence},
if $j = \bari$, $\tilde{p}_{ij} \geq w_1(i, j) + 2$ and, if $k = \barl$,
$\tilde{p}_{k\ell}\geq w_2(k, \ell) + 2 $.
Therefore, by definition of $\epsilon_{ij}$ and $\epsilon_{k\ell}$, we have
\begin{align*}
   w_2(i, j)
    &\geq \octproj_i(\vect{p}) - \octproj_j(\vect{p})\\
    &\geq w_1(i, j) + \epsilon_{ij}, \\
   w_1(k, \ell)
    &\geq \octproj_k(\vect{p}) - \octproj_\ell(\vect{p})\\
    &\geq w_2(k, \ell) + \epsilon_{k\ell}
\intertext{%
so that conditions (1a) and (1b) hold. Moreover,
}
  w(i, \ell) + w(k, j)
    &\geq \octproj_i(\vect{p}) - \octproj_\ell(\vect{p})
      + \octproj_k(\vect{p}) - \octproj_j(\vect{p}) \\
    &= \octproj_i(\vect{p}) - \octproj_j(\vect{p})
         + \octproj_k(\vect{p}) - \octproj_\ell(\vect{p}) \\
    &\geq w_1(i, j) + w_2(k, \ell) + \epsilon_{ij} + \epsilon_{k\ell}
\intertext{%
so that condition (2a) holds
and, by a symmetrical argument,
condition (2b) holds.
Similarly,
}
  w(i, \ell) + w(k, \bari) + w(\barj, j)
    &\geq \bigl( \octproj_i(\vect{p}) - \octproj_\ell(\vect{p}) \bigr)
      + \bigl( \octproj_k(\vect{p}) + \octproj_i(\vect{p}) \bigr)
      + \bigl( - 2 \octproj_j(\vect{p}) \bigr) \\
    &= 2 \bigl( \octproj_i(\vect{p}) - \octproj_j(\vect{p}) \bigr)
         + \octproj_k(\vect{p}) - \octproj_\ell(\vect{p}) \\
    &\geq 2 w_1(i, j) + w_2(k, \ell) + 2 \epsilon_{ij} + \epsilon_{k\ell}
\end{align*}
so that condition (3a) holds;
conditions (3b), (4a) and (4b)
follow by a symmetrical argument.

We now suppose that, for some $i, j, k, \ell \in \cN$,
conditions (1a) -- (4b) hold.
Consider the graph $G' = (\cN, w')$ where,
for each $r, s \in \cN$,
\begin{equation*}
  w'(r, s)
    \defeq
      \begin{cases}
        - w_1(i, j) - \epsilon_{ij},
          &\text{if $(r, s) \in \bigl\{(j, i), (\bari, \barj)\bigr\}$;} \\
        - w_2(k, \ell) - \epsilon_{k\ell},
          &\text{if $(r, s) \in \bigl\{(\ell, k), (\bark, \barl)\bigr\}$;} \\
        w(r, s),
          &\text{otherwise.}
      \end{cases}
\end{equation*}
Let $G \defeq G_1 \graphlub G_2$;
as $G$ is coherent, $G'$ is coherent too;
as $G$ is tightly closed, $G'$ satisfies property~\eqref{rule:tight-coherence}.
Hence it follows from Lemma~\ref{lem:LahiriM05-lemma-4} that
$G'$ is $\Zset$-consistent if it has no negative weight cycles.
By using a reasoning similar to that in the proof of
Theorem~\ref{thm:union-of-octs-iff-hull}, it can be seen that
there are no negative weight cycles in $G'$ so that
$G'$ is $\Zset$-consistent and $G' \graphleq G_1 \graphlub G_2$.
Let $\oct' \in \IOctshapes_n$ be represented by $G'$,
so that $\emptyset \neq \oct' \sseq \oct_1 \uplus \oct_2$.
Since $w'(j, i) + w_1(i, j) < 0$,
we obtain $\oct' \inters \oct_1 = \emptyset$;
since $w'(\ell, k) + w_2(k, \ell) < 0$,
we obtain $\oct' \inters \oct_2 = \emptyset$.
Hence, $\oct_1 \uplus \oct_2 \neq \oct_1 \union \oct_2$.
\qed
\end{proof}

Since the tight closure and tight reduction procedures are both
in $\bigO(n^3)$~\cite{BagnaraHZ09FMSD},
the exact join detection algorithm for integer octagonal shapes
has the same asymptotic worst-case complexity of all the corresponding
algorithms for the other weakly relational shapes.

\section{Conclusion and Future Work}
\label{sec:conclusion}

Several applications dealing with the synthesis, analysis, verification
and optimization of hardware and software systems make use of
numerical abstractions.  These are sets of geometrical objects
---with the structure of a bounded join-semilattice--- that are used to
approximate the numerical quantities occurring in such systems.
In order to improve the precision of the approximation, sets
of such objects are often considered and, to limit redundancy
and its negative effects, it is important to ``merge'' objects whose
lattice-theoretic join corresponds to their set-theoretic union.

For a wide range of numerical abstractions, we have presented results
that state the necessity and sufficiency of relatively simple conditions
for the equivalence between join and union.  These conditions immediately
suggest algorithms that solve the corresponding decision problem.
For the case of convex polyhedra, we improve upon one of the algorithms
presented in \cite{BemporadFT00TR,BemporadFT01} by defining an
algorithm with better worst-case complexity.  For all the other
considered numerical abstractions, we believe the present paper
is breaking new ground.  In particular, for the case of NNC convex
polyhedra, we show that dealing with non-closedness brings significant
extra complications.  For the other abstractions, the algorithms
we propose have worst-case complexities that, in a sense,
\emph{match} the complexity of the abstraction, something that cannot
be obtained, e.g., by applying an algorithm for general convex polyhedra
to octagonal shapes.

All the above mentioned algorithms have been implemented in the Parma
Polyhedra Library \cite{BagnaraHZ08SCP}.\footnote{The Parma Polyhedra Library
is free software distributed under the terms of the GNU General Public License.
See \url{http://www.cs.unipr.it/ppl/} for further details.}
Besides being made directly available to the client applications,
they are used internally in order to implement widening
operators over powerset domains \cite{BagnaraHZ06STTT}.
Our preliminary experimental evaluation, though not extensive,
showed the efficiency of the algorithms is good, also thanks
to a careful coding following the
\emph{``first fail'' principle}.\footnote{This is a heuristics
whereby, in the implementation of a predicate whose success depends
on the success of many tests, those that are most likely to fail
are tried first.}

In this paper we have studied exact join detection for the most
popular abstract domains.  However, due to the importance numerical
domains have in the synthesis, analysis, verification and optimization
of hardware and software systems, due to the need to face the
complexity/precision trade-off in an application-dependent way,
new domains are proposed on a regular basis.  The fact that they may
be not so popular today does not impede that they can prove their strength
tomorrow.  These domains include: the
\emph{two variables per linear inequality} abstract domain
\cite{Simon08,SimonKH02},
\emph{octahedra} \cite{ClarisoC04},
\emph{template polyhedra} \cite{SankaranarayananDI08}, and
\emph{pentagons} \cite{LogozzoF08}.
It will be interesting to study exact join detection for these and
other domains, the objective being the one of finding specializations
with a complexity that matches the ``inherent complexity'' of the domain.

Even though preliminary experimentation suggests that ---in practice,
at least for some applications \cite{BagnaraHZ06STTT,BultanGP99}---
pairwise joins allow the removal of most redundancies,
work is still needed in the definition of efficient algorithms to
decide the exactness of join for $k > 2$ objects.  Moreover, it would
be useful to develop heuristics to mitigate the combinatorial explosion
when attempting full redundancy removal from a set of $m$ objects,
as it is clearly impractical to invoke $2^m-m-1$ times the decision
algorithm on $k = 2$, \dots,~$m$ objects.

\paragraph*{Acknowledgments}
We are grateful to the participants of the \emph{Graphite Workshop}
(AMD's Lone Star Campus, Austin, Texas, November 16--17, 2008) for
stimulating us to add efficient exact join detection algorithms to the
Parma Polyhedra Library, something that led us to the research
described in this paper.

We are also indebted to the anonymous referees for their careful and
detailed reviews, which allowed us to significantly improve the paper.


\newcommand{\noopsort}[1]{}\hyphenation{ Ba-gna-ra Bie-li-ko-va Bruy-noo-ghe
  Common-Loops DeMich-iel Dober-kat Di-par-ti-men-to Er-vier Fa-la-schi
  Fell-eisen Gam-ma Gem-Stone Glan-ville Gold-in Goos-sens Graph-Trace
  Grim-shaw Her-men-e-gil-do Hoeks-ma Hor-o-witz Kam-i-ko Kenn-e-dy Kess-ler
  Lisp-edit Lu-ba-chev-sky Ma-te-ma-ti-ca Nich-o-las Obern-dorf Ohsen-doth
  Par-log Para-sight Pega-Sys Pren-tice Pu-ru-sho-tha-man Ra-guid-eau Rich-ard
  Roe-ver Ros-en-krantz Ru-dolph SIG-OA SIG-PLAN SIG-SOFT SMALL-TALK Schee-vel
  Schlotz-hauer Schwartz-bach Sieg-fried Small-talk Spring-er Stroh-meier
  Thing-Lab Zhong-xiu Zac-ca-gni-ni Zaf-fa-nel-la Zo-lo }

\end{document}